\documentclass[12pt,a4paper]{article}
\usepackage{amsfonts}
\usepackage{amsmath}
\usepackage[onehalfspacing]{setspace}

\newtheorem{theorem}{Theorem}

\newtheorem{corollary}[theorem]{Corollary}

\newtheorem{lemma}[theorem]{Lemma}

\newtheorem{proposition}[theorem]{Proposition}

\newenvironment{proof}[1][Proof]{\noindent\textbf{#1.} }{\ \rule{0.5em}{0.5em}}
\newtheorem{example}[theorem]{Example}
\let\oldexample\example
\renewcommand{\example}{\oldexample\normalfont}
\input{tcilatex}

\newcommand{\ignore}[1]{}
\begin{document}

\title{Approximate Revenue Maximization with Multiple Items\thanks{%
This version: September 2017. Previous versions: February 2012; April 2012 (%
\emph{arXiv} 1204.1846 and Center for Rationality DP-606); May 2014; March
2017. Research partially supported by a European Research Council Advanced
Investigator grant (Hart) and by an Israel Science Foundation grant (Nisan).
We thank Motty Perry and Phil Reny for introducing us to the subject and for
many helpful discussions, and the referees for their very careful reading
and useful comments. A\ presentation that covers some of this work is
available at \texttt{http://www.ma.huji.ac.il/hart/abs/2good-p.html}}}
\author{Sergiu Hart\thanks{%
The Hebrew University of Jerusalem (Federmann Center for the Study of
Rationality, Department of Economics, and Institute of Mathematics).\quad 
\emph{E-mail}: \texttt{hart@huji.ac.il} \quad \emph{Web site}: \texttt{%
http://www.ma.huji.ac.il/hart}} \and Noam Nisan\thanks{%
The Hebrew University of Jerusalem (Federmann Center for the Study of
Rationality, and School of Computer Science and Engineering), and Microsoft
Research. \emph{E-mail}: \texttt{noam@cs.huji.ac.il} \quad \emph{Web site}: 
\texttt{http://www.cs.huji.ac.il/\symbol{126}noam}}}
\maketitle

\begin{abstract}
Maximizing the revenue from selling \emph{more than one} good (or item) to a
single buyer is a notoriously difficult problem, in stark contrast to the
one-good case. For two goods, we show that simple \textquotedblleft
one-dimensional" mechanisms, such as selling the goods separately, \emph{%
guarantee at least }$73\%$\emph{\ }of the optimal revenue when the
valuations of the two goods are independent and identically distributed, and
at least $50\%$ when they are independent.

For the case of $k>2$ independent goods, we show that selling them
separately guarantees at least a $c/\log ^{2}k$ fraction of the optimal
revenue; and, for independent and identically distributed goods, we show
that selling them as one bundle guarantees at least a $c/\log k$ fraction of
the optimal revenue.

Additional results compare the revenues from the two simple mechanisms of
selling the goods separately and bundled, identify situations where bundling
is optimal, and extend the analysis to multiple buyers.
\end{abstract}

\tableofcontents

\def\@biblabel#1{#1\hfill}
\def\thebibliography#1{\section*{References}
\addcontentsline{toc}{section}{References}
\list
{}{
\labelwidth 0pt
\leftmargin 1.8em
\itemindent -1.8em
\usecounter{enumi}}
\def\newblock{\hskip .11em plus .33em minus .07em}
\sloppy\clubpenalty4000\widowpenalty4000
\sfcode`\.=1000\relax\def\baselinestretch{1}\large \normalsize}
\let\endthebibliography=\endlist%

\section{Introduction\label{s:intro}}

Suppose that a seller has one good (or \textquotedblleft item") to sell to a
single buyer whose willingness to pay (or \textquotedblleft value") for the
good is $x.$ While $x$ is known to the buyer, it is unknown to the seller,
who knows only its distribution (given by a cumulative distribution function 
$F$). If the seller offers to sell the good for a price $p$ then the
probability that the buyer will buy is $1-F(p)$, and the seller's revenue
will be $p\cdot (1-F(p))$. The seller will choose a price $p^{\ast }$ that
maximizes this expression.

This problem is the classic monopolist-pricing problem. Looking at it from
an auction point of view, one may ask whether there are mechanisms for
selling the good that yield a higher revenue. Such mechanisms could be
indirect, could offer different prices for different probabilities of
getting the good, and so on. Yet, the characterization of optimal mechanisms
of Myerson (1981) (see also Riley and Samuelson 1981 and Riley and
Zeckhauser 1983) concludes that the take-it-or-leave-it offer at the above
price $p^{\ast }$ yields the optimal revenue among \emph{all} mechanisms.
Even more, Myerson's result also applies when there are multiple buyers, in
which case $p^{\ast }$ would be the reserve price in a second-price auction.

Now suppose that the seller has two (different) goods that he wants to sell
to a single buyer. Furthermore, consider the simplest case where the buyer's
values for the two goods are independently and identically distributed
according to the distribution $F$ (\textquotedblleft i.i.d.-$F$%
\textquotedblright\ for short), and where, furthermore, his valuation is
additive: if the value of the first good is $y$ and that of the second is $z$%
, then the value of the \emph{bundle} consisting of both goods is\footnote{%
Our buyer's demand is thus \emph{not} limited to one good (as is the case in
some of the existing literature; see \textquotedblleft unit-demand" in
Section \ref{sus:intro-literature}).} $y+z$. It would seem that since the
two goods are completely independent of each other, then the best one should
be able to do is to sell each of them separately in the optimal way, and
thus extract exactly twice the revenue one would make from a single good.
Yet this turns out to be false.

\begin{example}
Consider the one-good distribution $F$ taking values $1$ and $2$, each with
probability $1/2$. Let us first look at selling a single good optimally: the
seller can either choose to price it at $1$, selling always\footnote{%
Since we maximize revenue we can assume without loss of generality that ties
are broken by the buyer in a way that maximizes the seller's revenue. This
\textquotedblleft seller-favorable" property can always be achieved by
appropriate small perturbations of the mechanism; for instance, by the
seller giving a small fixed proportional discount on all payments. See Hart
and Reny (2015a, Section 1.2, and Remark (a) after Corollary 18).} and
getting a revenue of $1$, or choose to price the good at $2$, selling it
with probability $1/2$, again obtaining an expected revenue of $1$, and so
the optimal revenue from a single good is $1$. Now consider the following
mechanism for selling both goods: bundle them together, and sell the bundle
for price $3$. The probability that the sum of the buyer's values for the
two goods is at least $3$ is $3/4$, and so the revenue is $3\cdot 3/4=2.25$%
---larger than the revenue of $2$ that is obtained by selling them
separately.
\end{example}

However, that is not always so: bundling may sometimes be worse than selling
the goods separately.

\begin{example}
Consider the one-good distribution $F$ taking values $0$ and $1$, each with
probability $1/2$. Selling the two goods separately yields a revenue of $1/2$
from each good (set the price at $1),$ and so $1$ in total, whereas the
revenue from selling the bundle is only $3/4$ (the optimal price for the
bundle is $1$).
\end{example}

In other cases neither selling separately nor bundling is optimal.

\begin{example}
Consider the one-good distribution $F$ taking the values $0,$ $1,$ and $2$,
each with probability $1/3$. The unique optimal mechanism for two such
i.i.d. goods turns out to be\footnote{%
For distributions with finite support, finding the optimal mechanism amounts
to solving a linear programming problem.} to offer to the buyer the choice
between any single good at price $2$ and the bundle of both goods at a
\textquotedblleft discount\textquotedblright\ price of $3$. This mechanism
gets a revenue of $13/9\approx 1.44,$ which is larger than the revenue of $%
4/3\approx 1.33$ obtained either from selling the two goods separately or
from selling them as a single bundle.
\end{example}

A similar situation obtains for the uniform distribution on $[0,1]$, for
which neither bundling nor selling separately is optimal (Manelli and
Vincent 2006). In still other cases the optimal mechanism is not even
deterministic and must offer lotteries for the goods. This happens for
instance in the following example, taken from Hart and Reny (2015a, Example
4).\footnote{%
Examples in which randomization increases the revenue appear in the
literature, starting with Thanassoulis (2004) in the somewhat different
setup of unit demand, and Manelli and Vincent (2006, 2007, 2012). See
Section \ref{sus:intro-literature}.}

\begin{example}
Consider the distribution taking the values $1$, $2,$ and $4$, with
probabilities $1/6$, $1/2$, and $1/3$, respectively. It turns out that the
unique optimal mechanism for two such i.i.d. goods offers the buyer a choice
of one of the following options: buying a lottery ticket that has price $1$
and gives the first good \emph{with probability }$1/2$, buying a similar
lottery ticket for good $2,$ buying the bundle of both goods for a price of $%
4$, and buying nothing (and paying nothing); indeed, any deterministic
mechanism has a strictly lower revenue.
\end{example}

Thus, it is not clear what optimal mechanisms for selling two goods look
like, and indeed characterizations of optimal mechanisms even for this
simple case are not known (see Section \ref{sus:intro-literature}). The
two-dimensional problem is extremely difficult, and the simple mechanisms
that amount to solving only one-dimensional problems---such as separate
selling and bundling---do not maximize the revenue in general.

This leads to the following question: how good are such simple mechanisms
for selling two goods? That is, how much of the optimal revenue is
guaranteed when using them? Consider a class of mechanisms $\mathcal{N}$
(such as separate selling or bundled selling) and a class of environments $%
\mathbb{X}$ (such as two independent and identically distributed goods, or $%
k $ independent goods); we then define the\textbf{\ }\emph{Guaranteed
Fraction of Optimal Revenue}\textbf{\ }(\textsc{GFOR}) as that maximal
fraction $\alpha $ between $0$ and $1$ such that for every environment in $%
\mathbb{X}$ there is a mechanism in $\mathcal{N}$ that yields a revenue of
at least the fraction $\alpha $ of the optimal revenue (it is thus the
reciprocal of the \textquotedblleft competitive ratio" often used in the
computer science literature; see Section \ref{sus:gfor} for the formal
definition and a discussion of these concepts).

We start with two independent goods, and consider selling them separately.%
\emph{\ }Our first result is:%
\renewcommand{\thetheorem}{\Alph{theorem}}
\setcounter{theorem}{0}%

\begin{theorem}
\label{th:1/2}For any two independent goods, selling each good separately at
its optimal one-good price guarantees at least $50\%$ of the optimal
revenue; i.e., 
\begin{equation*}
\text{\textsc{GFOR}}(\text{\textsc{separate}};~2\text{ independent goods}%
)\geq \frac{1}{2}.
\end{equation*}
\end{theorem}

This result applies to any distribution of values (we make no assumptions,
such as monotone hazard rate or increasing virtual values), and it holds
also for any number of buyers (see Section \ref{sus:indep-n-buyers}).

When the two goods are identically distributed, separate selling is
guaranteed to perform even better.

\begin{theorem}
\label{th:73}For any two independent and identically distributed goods,
selling each one at the one-good optimal price guarantees at least $73\%$ of
the optimal revenue; i.e., 
\begin{equation*}
\text{\textsc{GFOR}}(\text{\textsc{separate}};~2\text{ i.i.d. goods})\geq 
\frac{e}{e+1}\approx 0.73.
\end{equation*}
\end{theorem}

Thus, for two i.i.d. goods with distribution $F$, setting the price at $%
p^{\ast }$ that maximizes the one-good revenue (i.e., $p^{\ast }(1-F(p^{\ast
}))=\max_{p}p(1-F(p)))$ and allowing the buyer to buy any number of units---$%
0,1,$ or $2$ units---at price $p^{\ast }$ per unit guarantees at least $73\%$
of the optimal revenue.

We next consider the case of more than two goods. It turns out that, as the
number $k$ of goods grows, the fraction of the optimal revenue that is
obtainable from selling them separately may become arbitrarily small
(specifically, of the order of $1/\log k,$ cf. Corollary \ref{c:sep-k-er};
the reader may refer to the tables in Appendix \ref{ap:summary} that
summarize all these comparisons). Our main positive result here is:

\begin{theorem}
\label{th:srev-k}There exists a constant $c>0$ such that for any $k\geq 2$
and any $k$ independent goods, selling each good separately at its optimal
one-good price guarantees at least $c/\log ^{2}k$ of the optimal revenue;
i.e.,%
\begin{equation*}
\text{\textsc{GFOR}}(\text{\textsc{separate}};~k\text{ independent goods}%
)\geq \frac{c}{\log ^{2}k}.
\end{equation*}
\end{theorem}

Finally, we move to the other simple one-dimensional mechanism, the \emph{%
bundling} mechanism, which offers a single price for the bundle of all
goods. We first show that for general independent goods, bundling may do
much worse and yield only a $1/k$-fraction of the optimal revenue (Example %
\ref{ex:1/k}). However, when the goods are independent and \emph{identically}
\emph{distributed}, then bundling does much better. It is well known
(Armstrong 1999, Bakos and Brynjolfsson 1999) that for every fixed
distribution $F$, as the number of goods distributed independently according
to $F$ increases, the bundling mechanisms become close to being optimal (for
completeness we provide a short proof in Appendix \ref{s:k-iid}). This,
however, requires $k$ to grow as $F$ remains fixed. On the other hand, we
show that this is not true uniformly over $F$: for every large enough $k$,
there are distributions where the bundling mechanism on $k$ goods gives less
than $57\%$ of the optimal revenue (Example \ref{ex:57}). Our main result
for the bundling mechanism in the i.i.d. case is:

\begin{theorem}
\label{th:brev-k} There exists a constant $c>0$ such that for any $k\geq 2$
and any $k$ independent and identically distributed goods, selling them as
one bundle at the bundle-optimal price guarantees at least $c/\log k$ of the
optimal revenue:%
\begin{equation*}
\text{\textsc{GFOR}}(\text{\textsc{bundled}};~k\text{ i.i.d. goods})\geq 
\frac{c}{\log k}.
\end{equation*}
\end{theorem}

\renewcommand{\thetheorem}{\arabic{theorem}}%
The paper is organized as follows. Section \ref{sus:intro-literature}
presents a survey of the relevant literature, including work done following
the circulation of the early versions of this paper in 2012. In Section \ref%
{s:prelim} we present the model and the basic concepts, followed by a number
of useful preliminary results. Section \ref{s:independent} deals with the
case of two independent goods, and provides the proof of Theorem \ref{th:1/2}%
; the more complex proof of Theorem \ref{th:73} is relegated to Appendix \ref%
{s:proof-iid}. The main argument of these proofs is then extended to a
general decomposition theorem in Section \ref{s:decomposition}. Section \ref%
{s:sep-bun} studies the relations between the revenues from separate and
bundled selling (with some of the proofs and additional results in
Appendices \ref{s:ER} and \ref{s:ap-sep-bun}); these relations are not only
interesting in their own right, but are also used as part of the general
analysis, and provide us with most of the examples that we have of gaps in
revenue. The results for more than two goods, Theorems \ref{th:srev-k} and %
\ref{th:brev-k}, are then proved in Section \ref{s:proofs-thm3-4}, making
use of the decomposition of Section \ref{s:decomposition} and the
comparisons of Section \ref{s:sep-bun}. Additional relevant results, namely,
upper bounds on \textsc{GFOR}, a two-good setup where bundling is shown to
be optimal, and the extension to multiple buyers, are stated in Section \ref%
{s:additional} (with proofs relegated to Appendices \ref{s:brev=rev} and \ref%
{s:n-buyers-proofs}). In Section \ref{s:discussion} we discuss open
problems. Finally, Appendix \ref{ap:summary} provides two tables: the first
summarizes the lower and upper bounds that we have obtained on the fraction
of optimal revenue that is guaranteed for both separate and bundled selling,
and the second summarizes the comparisons between separate and bundled
selling.

\subsection{Literature\label{sus:intro-literature}}

We briefly describe some of the existing work on these issues. McAfee and
McMillan (1988) identify cases where the optimal mechanism is deterministic.
However, Thanassoulis (2004) and Manelli and Vincent (2006) found a
technical error in the paper and present counterexamples. These last two
papers contain good surveys of the work within economic theory, with more
recent analysis by Fang and Norman (2006), Jehiel, Meyer-ter-Vehn, and
Moldovanu (2007), Pycia (2006), Lev (2011), Pavlov (2011), and Hart and Reny
(2015a). In the past few years algorithmic work on these types of topics was
carried out. One line of work (e.g., Briest, Chawla, Kleinberg, and Weinberg
2015; Cai, Daskalakis, and Weinberg 2012a; Alaei, Fu, Haghpanah, Hartline,
and Malekian 2012) shows that for discrete distributions the optimal
mechanism can be found by linear programming in rather general settings.
This is certainly true in our simple setting where the direct representation
of the mechanism constraints provides a polynomial-size linear program. Thus
we emphasize that the difficulty in our case is \emph{not computational},
but is rather one of characterizing and understanding the results of the
explicit computations: this is certainly so for continuous distributions,
but also for discrete ones.\footnote{\label{ftn:conceptual-complexity}This
suggests that a notion of \emph{\textquotedblleft conceptual complexity"}
may be appropriate here. The usual computational complexity may not capture
all the difficulty of a problem, since, even after computing the precise
solution, one may not understand its structure, what it means and
represents, and how it varies with the given parameters (e.g., Hart and Reny
2015a, where it is shown that the optimal revenue may \emph{decrease }when
the buyer's valuations \emph{increase}).} Another line of work in computer
science (Chawla, Hartline, and Kleinberg 2007; Chawla, Hartline, Malec, and
Sivan 2010; Chawla, Malec, and Sivan 2010; Alaei, Fu, Haghpanah, Hartline,
and Malekian 2012; Cai, Daskalakis, and Weinberg 2012b) attempts to
approximate the optimal revenue by simple mechanisms. This was done for
various settings, especially unit-demand settings and some generalizations.%
\footnote{%
\textquotedblleft Unit demand" means that buyers are willing to buy \emph{at
most} one of the goods. The relations between the revenues in the additive
setup and those in the unit-demand setup are discussed in our paper Hart and
Nisan (2013, Appendix 1). It is thus possible that bounds obtained in the
unit-demand literature lead, in the additive setup, to weaker versions of
our Theorem \ref{th:1/2}. Our direct approach is simpler---cf. Section \ref%
{s:independent}---and generalizable---cf. Section \ref{s:decomposition}.}
One particular conclusion from this line of work is that for many subclasses
of distributions (such as those with a monotone hazard rate) various simple
mechanisms can extract a constant fraction of the expected value of the
goods.\footnote{%
In our setting this is true even more generally, for instance, whenever the
ratio between the median and the expectation is bounded (which happens in
particular when the tail of the distribution is \textquotedblleft
thinner\textquotedblright\ than $x^{-\alpha }$ for $\alpha >1$). Indeed,
posting a price equal to the median yields a revenue of one-half of the
median, and hence at least a constant fraction of the expectation (which is,
by IR, the most that the seller can extract as expected revenue).} This is
true in our simple setting, where for such distributions selling the goods
separately provides a constant fraction of the expected value and thus of
the optimal revenue. The case of multiple goods with \emph{correlated
distributions} was studied by Briest, Chawla, Kleinberg, and Weinberg (2010)
and Hart and Nisan (2013) and turns out to be quite different from the
independently distributed case: classes of simple mechanisms (even the class
of all deterministic mechanisms) may well yield only an arbitrarily small
fraction of the optimal revenue.

Since the circulation of early versions of this paper in 2012 (Hart and
Nisan 2012), there has been a flurry of work on optimal mechanisms for
multiple goods for various cases: Daskalakis, Deckelbaum, and Tzamos (2013,
2017), Giannakopoulos (2014), Giannakopoulos and Koutsoupias (2014),
Menicucci, Hurkens, and Jeon (2015), and Tang and Wang (2017). The general
problem was studied from a computational perspective in Daskalakis,
Deckelbaum, and Tzamos (2014), where it is shown to be computationally
intractable (formally, $\#P$-hard). Several developments have occurred
regarding \textsc{GFOR} for multiple goods. Li and Yao (2013) improved our
lower bound on \textsc{GFOR(separate)} for $k$ goods from $c/\log ^{2}k$ to
the tight $c/\log k$. For the case of $k$ independent and identically
distributed goods, Li and Yao (2013) proved that \textsc{GFOR(bundled)} is
bounded from below by a constant that is independent of the number of goods $%
k.$ Babaioff, Immorlica, Lucier, and Weinberg (2014) showed that, for $k$
independent (but not necessarily identically distributed) goods, \textsc{GFOR%
}$(\{$\textsc{separate},\textsc{~bundled}$\})$ is bounded from below by a
constant that is~independent of $k$ (that is, there is $c>0$ such that for
any number $k$ and any $k$ independent goods, either separate selling or
bundling yields at least the fraction $c$ of the optimal revenue). This was
generalized by Yao (2014) to the case of multiple bidders and by Rubinstein
and Weinberg (2015) to buyers with submodular (rather than just additive)
valuations for the goods. Measures quantifying how complex mechanisms need
to be in order to yield a good proportion of the optimal revenue were
studied in Hart and Nisan (2013), Dughmi, Han, and Nisan (2014), Morgenstern
and Roughgarden (2016), and Babaioff, Gonczarowski, and Nisan (2017).

\section{Preliminaries\label{s:prelim}}

In this section we present the model formally and define the concepts that
we use, followed by a number of preliminary results.

\subsection{The Model\label{sus:model}}

One seller (or \textquotedblleft monopolist") is selling a number $k\geq 1$
of goods\ (or \textquotedblleft items," \textquotedblleft objects," etc.) to
one buyer.

The goods have no value or cost to the seller. Let $x_{1},x_{2},...,x_{k}%
\geq 0$ be the buyer's values for the goods. The value for getting a set of
goods is \emph{additive}: getting the subset $I\subseteq \{1,2,...,k\}$ of
goods is worth $\sum_{i\in I}x_{i}$ to the buyer (and so, in particular, the
buyer's demand is \emph{not} restricted to one good only). The values are
given by a random variable $X=(X_{1},X_{2},...,X_{k})$ that takes values in $%
\mathbb{R}_{+,}^{k}$ (we thus assume that valuations are always
nonnegative); we will refer to $X$ as a $k$\emph{-good} \emph{random
valuation.} The realization $x=(x_{1},x_{2},...,x_{k})\in \mathbb{R}_{+}^{k}$
of $X$ is known to the buyer, but not to the seller, who knows only the
distribution $F$ of $X$ (which may be viewed as the seller's belief). The
buyer and the seller are assumed to be risk neutral and to have quasi-linear
utilities (i.e., the utility is additive with respect to monetary transfers;
e.g., getting good $1$ with probability $1/2$ and paying $8$ with
probability $1/4$ is worth $(1/2)\cdot x_{1}-(1/4)\cdot 8$ to the buyer and $%
(1/4)\cdot 8$ to the seller).

The objective is to \emph{maximize} the seller's (expected) \emph{revenue}.

As has been well established by the so-called \textquotedblleft Revelation
Principle\textquotedblright\ (starting with Myerson 1981; see for instance
the book of Krishna 2010), we can restrict ourselves to \textquotedblleft
direct mechanisms" and \textquotedblleft truthful
equilibria.\textquotedblright\ A (direct)\footnote{%
\textquotedblleft Mechanism\textquotedblright\ will henceforth always mean
\textquotedblleft direct mechanism.\textquotedblright}\emph{\ mechanism} $%
\mu $ consists of a pair of functions $(q,s),$ where $%
q=(q_{1},q_{2},...,q_{k}):\mathbb{R}_{+}^{k}\rightarrow \lbrack 0,1]^{k}$
and $s:\mathbb{R}_{+}^{k}\rightarrow \mathbb{R},$ which prescribe the \emph{%
allocation} of goods and the \emph{payment}, respectively. Specifically, if
the buyer reports a value vector $x\in \mathbb{R}_{+}^{k},$ then $%
q_{i}(x)\in \lbrack 0,1]$ is the probability that the buyer receives good%
\footnote{%
When the goods are infinitely divisible and the valuations are linear in
quantities (i.e., the value of a quantity $\lambda $ of good $i$ is $\lambda
x_{i}),$ we may interpret $q_{i}$ also as the \emph{quantity} of good $i$
that the buyer gets.} $i$ (for $i=1,2,...,k$), and $s(x)$ is the payment
that the seller receives from the buyer. When the buyer reports his value $x$
truthfully, his payoff is\footnote{%
When $y=(y_{i})_{i=1,...,n}$ and $z=(z_{i})_{i=1,...,n}$ are $n$-dimensional
vectors, $y\cdot z$ denotes their scalar product $\sum_{i=1}^{n}y_{i}z_{i}.$}
$b(x)=\sum_{i=1}^{k}q_{i}(x)x_{i}-s(x)=q(x)\cdot x-s(x),$ and the seller's
payoff is\footnote{%
In the literature the payment to the seller is called transfer, cost, price,
revenue, and so on, and is denoted by $t,c,p,...;$ this plethora of names
and notations applies to the buyer as well. We hope that using the mnemonic $%
s$ for the seller's final payoff and $b$ for the buyer's final payoff will
avoid confusion.} $s(x).$ The mechanism $\mu =(q,s)$ satisfies \emph{%
individual rationality} (\textbf{IR)} if $b(x)\geq 0$ for every $x\in 
\mathbb{R}_{+}^{k},$ and \emph{incentive compatibility} (\textbf{IC}) if $%
b(x)\geq q(\tilde{x})\cdot x-s(\tilde{x})$ for every alternative report $%
\tilde{x}\in \mathbb{R}_{+}^{k}$ of the buyer when his value is $x,$ for
every $x\in \mathbb{R}_{+}^{k}.$ Let $\mathcal{M}$ denote the class of all
IC and IR mechanisms $\mu =(q,s).$ The expected revenue from a buyer with
random valuation $X$ using a mechanism $\mu =(q,s)\in \mathcal{M}$ is%
\footnote{%
In Hart and Reny (2015a, Proposition 16) it is shown that for IC and IR
mechanisms one may assume without loss of generality that $s$ is measurable;
since $s$ is bounded from below by $s(0)$ (which follows from IC at $0),$
the expected revenue $\mathbb{E}\left[ s(X)\right] $ is well defined (but
may be infinite).} $R(\mu ;X):=\mathbb{E}\left[ s(X)\right] ,$ and the \emph{%
optimal revenue} from $X$ is, by the Revelation Principle, \textsc{Rev}$%
(X):=\sup_{\mu \in \mathcal{M}}R(\mu ;X),$ the highest revenue that can be
obtained by any IC and IR mechanism $\mu .$ The revenue can never exceed the
expected valuation of all goods together: \textsc{Rev}$(X)\leq \mathbb{E}%
\left[ \sum_{i}X_{i}\right] $ (since $s(x)\leq q(x)\cdot x\leq \sum_{i}x_{i}$
by IR and $x\geq 0$).

When there is only one good, i.e., when $k=1,$ Myerson's (1981) result is
that 
\begin{equation}
\text{\textsc{Rev}}(X)=\sup_{p\geq 0}p\cdot \mathbb{P}\left[ X\geq p\right]
=\sup_{p\geq 0}p\cdot \mathbb{P}\left[ X>p\right] =\sup_{p\geq 0}p\cdot
(1-F(p)),  \label{eq:one good}
\end{equation}%
where $F$ is the cumulative distribution function of $X.$ Optimal mechanisms
correspond to the seller \textquotedblleft posting" a price $p$ and the
buyer buying the good for the price $p$ whenever his value is at least $p$;
in other words, the seller makes the buyer a \textquotedblleft
take-it-or-leave-it" offer to buy the good at price $p.$

Besides the maximal revenue, we are also interested in what can be obtained
from certain classes of mechanisms. Thus, given a class $\mathcal{N}\subset 
\mathcal{M}$ of IC and IR mechanisms$,$ let $\mathcal{N}$-\textsc{Rev}$%
(X):=\sup_{\nu \in \mathcal{N}}R(\nu ;X)$ be the maximal revenue that can be
extracted from a buyer with random valuation $X$ when restricted to
mechanisms $\nu $ in the class $\mathcal{N}.$ Some classes of mechanisms are:

\begin{itemize}
\item \textsc{Separate}: Each good $i$ is sold separately. The maximal
revenue from separate mechanisms is denoted by \textsc{SRev}$,$ and so 
\begin{equation*}
\text{\textsc{SRev}}(X):=\text{\textsc{Rev}}(X_{1})+\text{\textsc{Rev}}%
(X_{2})+...+\text{\textsc{Rev}}(X_{k}).
\end{equation*}

\item \textsc{Bundled}: All goods are sold together in one \textquotedblleft
bundle." The maximal revenue from bundled mechanisms is denoted by \textsc{%
BRev}, and so 
\begin{equation*}
\text{\textsc{BRev}}(X):=\text{\textsc{Rev}}(X_{1}+X_{2}+...+X_{k}).
\end{equation*}

\item \textsc{Deterministic:} Each good $i$ is either fully allocated or not
at all, i.e., $q_{i}(x)\in \{0,1\}$ (rather than $q_{i}(x)\in \lbrack 0,1])$
for every $x\in \mathbb{R}_{+}^{k}$ and $1\leq i\leq k.$ The maximal revenue
from deterministic mechanisms is denoted by \textsc{DRev}.
\end{itemize}

\noindent The separate and the bundled revenues are obtained by solving
one-dimensional problems (where one uses (\ref{eq:one good})), whereas the
deterministic revenue is a multi-dimensional problem. Examples where these
classes yield revenues that are smaller than the maximal revenue are well
known (see also the examples in the Introduction, and the references in
Section \ref{sus:intro-literature}).

The following is a useful characterization of incentive compatibility that
is well known (starting with Rochet 1985).

\begin{proposition}
\label{p:IC =b-convex}Let $\mu =(q,s)$ be a mechanism for $k$ goods with
buyer payoff function $b.$ Then $\mu =(q,s)$ satisfies IC if and only if $b$
is a convex function and for all $x$ the vector $q(x)$ is a subgradient of $%
b $ at $x$ (i.e., $b(\tilde{x})-b(x)\geq q(x)\cdot (\tilde{x}-x)$ for all $%
\tilde{x})$.
\end{proposition}

\begin{proof}
$\mu $ is IC if and only if $b(x)=q(x)\cdot x-s(x)=\max_{\tilde{x}\in 
\mathbb{R}_{+}^{k}}(q(\tilde{x})\cdot x-s(\tilde{x}))$ for every $x,$ which
implies that $b$ is a convex function of $x$ (as the maximum of a collection
of affine functions of $x).$ Moreover, for every $x$ and $\tilde{x}$ we have 
$b(\tilde{x})-b(x)-q(x)\cdot (\tilde{x}-x)=b(\tilde{x})-(q(x)\cdot \tilde{x}%
-s(x)),$ and so the subgradient inequalities are precisely the IC
inequalities.
\end{proof}

\bigskip

Thus $s(x)=q(x)\cdot x-b(x)=\nabla b(x)\cdot x-b(x),$ where $\nabla b(x)$
stands for a (sub)gradient of $b$ at $x,$ and so the revenue can be
expressed in terms of the buyer payoff function\footnote{%
The function $b,$ being convex, is differentiable almost everywhere, and so $%
\nabla b(x)$ is the gradient $(\partial b(x)/\partial x_{i})_{i=1,...,k}$
for almost every $x.$ As pointed out in Hart and Reny (2015a, Appendix A.1),
when maximizing revenue one may use \textquotedblleft seller-favorable"
mechanisms and replace the term $\nabla b(x)\cdot x$ with $b^{\prime }(x;x),$
the directional derivative of $b$ at $x$ in the direction $x,$ which is well
defined for every $x.$} $b.$ We also note that there is no loss of
generality in assuming that the mechanism $\mu $ is defined and satisfies IC
and IR on the whole space $\mathbb{R}_{+}^{k},$ rather than just on some
domain $D\subset \mathbb{R}_{+}^{k},$ such as the set of possible values of $%
X;$ see Hart and Reny (2015a, Appendix A.1).

We conclude with a useful property: a mechanism $\mu =(q,s)$ satisfies the 
\emph{no positive transfer}\footnote{%
The \textquotedblleft transfer" is from the seller to the buyer, i.e., $%
-s(x).$} (\textbf{NPT}) property if $s(x)\geq 0$ for every $x\in \mathbb{R}%
_{+}^{k}$. Proposition \ref{p:sub-dom} below shows that NPT can always be
assumed without loss of generality when maximizing revenue.\footnote{%
This is not true in more general setups; for instance, when there are
multiple buyers that are correlated, Bayesian Nash implementation may
require using positive transfers, i.e., $s(x)<0$ (cf. Cr\'{e}mer and McLean
1988; see also Appendix \ref{s:n-buyers-proofs} below).} Moreover, the
revenue from any IC and IR mechanism on a \emph{subdomain} of valuations
cannot exceed the overall maximal revenue---even when the mechanism does not
satisfy NPT.\footnote{%
This is used in our proofs, e.g., in Section \ref{s:independent}, where we
construct mechanisms for which $s$ may take negative values.}

\begin{proposition}
\label{p:sub-dom}Let $\mu =(q,s)$ be an IC and IR mechanism, and let $X$ be
a $k$-good random valuation in $\mathbb{R}_{+}^{k}$, where $k\geq 1.$ Then:

\begin{description}
\item[(i)] $\mu $ satisfies NPT if and only if $s(0)=0,$ which occurs if and
only if $b(0)=0.$

\item[(ii)] There is a mechanism $\hat{\mu}=(q,\hat{s})$ with the same $q$
and with $\hat{s}(x)\geq s(x)$ for all $x\in \mathbb{R}_{+}^{k},$ such that $%
\hat{\mu}$ satisfies IC, IR, and NPT.

\item[(iii)] \textsc{Rev}$(X)=\sup_{\mu }R(\mu ;X)$ where the supremum is
taken over all IC, IR, and NPT mechanisms $\mu $.

\item[(iv)] Let $A\subseteq \mathbb{R}_{+}^{k}$ be a set of values of $X;$
then\footnote{%
We write $\mathbf{1}_{W}$ for the indicator of the event $W$: it takes the
value $1$ when $W$ occurs and the value $0$ otherwise.} 
\begin{equation*}
\mathbb{E}\left[ s(X)\,\mathbf{1}_{X\in A}\right] \leq \text{\textsc{Rev}}%
(X\,\mathbf{1}_{X\in A})\leq \text{\textsc{Rev}}(X).
\end{equation*}
\end{description}
\end{proposition}

\begin{proof}
\textbf{(i)} IC at $0$ yields $s(x)\geq s(0)$ for all $x,$ and IR at $0$
yields $s(0)\leq 0;$ the minimal payment is thus $s(0),$ which cannot be
positive. Therefore, $s(x)\geq 0$ for all $x$ if and only if $s(0)=0.$ Now $%
s(0)+b(0)=q(0)\cdot 0=0,$ and so $s(0)=0$ if and only if $b(0)=0$.

\textbf{(ii)} Put $\hat{s}(x):=s(x)-s(0)\geq s(x)$ for all $x$ (recall that $%
s(0)\leq 0$ by IR at $0).$ Then $\hat{\mu}=(q,\hat{s})$ satisfies IC since
the payment differences have not changed (i.e., $\hat{s}(x)-\hat{s}(\tilde{x}%
)=s(x)-s(\tilde{x})$ for all $x,\tilde{x})$; it satisfies IR since $%
q(x)\cdot x-s(x)\geq q(0)\cdot x-s(0)\geq -s(0)$ (by IC); and it satisfies
NPT since $\hat{s}(0)=0.$

\textbf{(iii)} Follows from (ii) since $\hat{\mu}$ yields at least as much
revenue as $\mu $ (because $\hat{s}(x)\geq s(x)$ for all $x).$

\textbf{(iv)} For the first inequality, use (ii) to get $\mathbb{E}\left[
s(X)\,\mathbf{1}_{X\in A}\right] \leq \mathbb{E}\left[ \hat{s}(X)\,\mathbf{1}%
_{X\in A}\right] =\mathbb{E}\left[ \hat{s}(X\,\mathbf{1}_{X\in A})\right]
\leq $\textsc{Rev}$(X\,\mathbf{1}_{X\in A})$ (the equality since $\hat{s}%
(0)=0$ by (i)). For the second inequality, $\mathbb{E}\left[ s(X\,\mathbf{1}%
_{X\in A})\right] =\mathbb{E}\left[ s(X)\,\mathbf{1}_{X\in A}\right] \leq 
\mathbb{E}\left[ s(X)\right] $ for any $\mu $ that satisfies NPT; apply
(iii).
\end{proof}

\subsection{Guaranteed Fraction of Optimal Revenue (GFOR)\label{sus:gfor}}

Let $\mathbb{X}$ be a class of random valuations (such as two independent
goods, or $k$ i.i.d. goods; formally, it is a class of random variables $X$
with values in $\mathbb{R}_{+}^{k}$ spaces), and let $\mathcal{N}$ be a
class of IC and IR mechanisms (such as separate selling, or deterministic
mechanisms; formally, $\mathcal{N}$ is a subset of the class $\mathcal{M}$
of all IC and IR mechanisms). The \emph{Guaranteed Fraction of Optimal
Revenue }(\textsc{GFOR}) for the class of random valuations $\mathbb{X}$ and
the class of mechanisms $\mathcal{N}$ is defined as the maximal fraction $%
\alpha $ such that, for any random valuation $X$ in $\mathbb{X},$ there are
mechanisms in the class $\mathcal{N}$ that yield at least the fraction $%
\alpha $ of the optimal revenue. Formally,\footnote{%
Put $0/0=1.$}%
\begin{equation*}
\text{\textsc{GFOR}}\equiv \text{\textsc{GFOR}}(\mathcal{N};\mathbb{X}%
):=\inf_{X\in \mathbb{X}}\frac{\mathcal{N}\text{-\textsc{Rev}}(X)}{\text{%
\textsc{Rev}}(X)}=\inf_{X\in \mathbb{X}}\frac{\sup_{\nu \in \mathcal{N}%
}R(\nu ;X)}{\sup_{\mu \in \mathcal{M}}R(\mu ;X)}.
\end{equation*}

Thus \textsc{GFOR\thinspace }$\geq \alpha $ if and only if for every random
valuation $X$ in $\mathbb{X}$ there is a mechanism $\nu $ in $\mathcal{N}$
such that its revenue is $R(\nu ;X)\geq \alpha \cdot $\textsc{Rev}$(X)$ (we
are ignoring here the trivial issues of \textquotedblleft
max\textquotedblright\ vs. \textquotedblleft sup\textquotedblright ), and 
\textsc{GFOR\thinspace }$\leq \alpha $ if there exists a random valuation $X$
in $\mathbb{X}$ such that for every mechanism $\nu $ in $\mathcal{N}$ its
revenue is $R(\nu ;X)\leq \alpha \cdot $\textsc{Rev}$(X).$

\bigskip

\noindent \textbf{Remarks.} \emph{(a) }One may argue that there is no need
for results that are \emph{uniform} with respect to the values'
distributions, on the grounds that the seller knows that distribution.
However, in the case of multiple goods, knowing the distribution does not
help find the optimal mechanism (even for simple distributions), whereas
simple mechanisms, such as separate selling, are always easy to compute (as
they use only optimal prices for one-dimensional distributions). It is thus
important to know how far from optimal these mechanisms are guaranteed to
be, particularly when one does not know what that optimum is or how to find
it.

\emph{(b) Ratios.} Why are we considering ratios? The reason is that the
revenue is covariant with rescalings, but not with translations. Indeed, 
\textsc{Rev}$(\lambda X)=\lambda \cdot $\textsc{Rev}$(X)$ for any $\lambda
>0,$ but \textsc{Rev}$(X+c)$ is in general different from \textsc{Rev}$(X)+c$
for constant $c>0$ (this happens already in the one-good case; see (\ref%
{eq:one good})).

\emph{(c) Competitive ratio}. The computer science literature uses the
concepts of \textquotedblleft competitive ratio\textquotedblright\ and
\textquotedblleft approximation ratio," which are just the reciprocal $1/$%
\textsc{GFOR }of\textsc{\ GFOR}. While the two notions are clearly
equivalent, using the optimal revenue as the benchmark (i.e., $100\%$) and
measuring everything relative to this basis---as \textsc{GFOR} does---seems
to come more naturally.

\section{Two Independent Goods\label{s:independent}}

We start by proving our first result, Theorem \ref{th:1/2}, stated in the
Introduction. Its proof forms the basis of more complex proofs
later---including Theorem \ref{th:73}, whose significantly more intricate
proof is relegated to Appendix \ref{s:proof-iid}. Theorem \ref{th:1/2} can
be restated as follows: for every two-good random valuation $X=(Y,Z)$\ with $%
Y,Z$ independent\ goods (i.e., one-dimensional nonnegative random variables),%
\begin{equation}
\text{\textsc{Rev}}(X)\leq 2\cdot \text{\textsc{SRev}}(X)=2(\text{\textsc{Rev%
}}(Y)+\text{\textsc{Rev}}(Z)).  \label{eq:2+2}
\end{equation}

\bigskip

\begin{proof}[Proof of Theorem \protect\ref{th:1/2}]
Let $\mu =(q,s)$ be a two-good IC, IR, and NPT mechanism (recall Proposition %
\ref{p:sub-dom} (iii)); we will prove that its revenue from $X$ satisfies $%
R(\mu ;X)\leq 2$\textsc{Rev}$(Y)+2$\textsc{Rev}$(Z)$. To do so, we split the
revenue into two parts, according to which one of $Y$ and $Z$ is higher, and
show that%
\begin{eqnarray}
\mathbb{E}\left[ s(Y,Z)\,\mathbf{1}_{Y\geq Z}\right] &\leq &2\text{\textsc{%
Rev}}(Y),\text{\ \ and}  \label{eq:triangle} \\
\mathbb{E}\left[ s(Y,Z)\,\mathbf{1}_{Z\geq Y}\right] &\leq &2\text{\textsc{%
Rev}}(Z).  \label{eq:triangle2}
\end{eqnarray}%
Since $R(\mu ;X)=\mathbb{E}\left[ s(Y,Z)\right] \leq \mathbb{E}\left[
s(Y,Z)\,\mathbf{1}_{Y\geq Z}\right] +\mathbb{E}\left[ s(Y,Z)\,\mathbf{1}%
_{Z\geq Y}\right] $ (the inequality is due to the diagonal $Y=Z$ being
counted twice; recall that $s\geq 0$ by NPT), adding (\ref{eq:triangle}) and
(\ref{eq:triangle2}) gives (\ref{eq:2+2}).

We now prove (\ref{eq:triangle}) (which then yields (\ref{eq:triangle2}) by
interchanging $Y$ and $Z$). For every fixed value $z\geq 0$ of the second
good define a mechanism $\mu ^{z}=(q^{z},s^{z})$ for the first good by
replacing the allocation of the second good with an equivalent decrease in
payment; that is, the allocation of the first good is unchanged, i.e., $%
q^{z}(y):=q_{1}(y,z),$ and the payment is $s^{z}(y):=s(y,z)-q_{2}(y,z)\cdot
z,$ for every $y\geq 0.$ The one-good mechanism $\mu ^{z}$ is IC and IR for $%
y$, since $\mu =(q,s)$ was IC and IR for $(y,z)$ (for IC, only the
constraints $(\tilde{y},z)$ vs. $(y,z)$ matter; for IR, the buyer payoff
function of $\mu ^{z}$ is $b^{z}(y)=b(y,z)$). Now $%
s(y,z)=s^{z}(y)+q_{2}(y,z)\cdot z\leq s^{z}(y)+z$ (because $z\geq 0$ and $%
q_{2}\leq 1),$ and so%
\begin{equation*}
\mathbb{E[}s(Y,Z)\,\mathbf{1}_{Y\geq Z}~|~Z=z]=\mathbb{E}\left[ s(Y,z)\,%
\mathbf{1}_{Y\geq z}\right] \leq \mathbb{E}\left[ s^{z}(Y)\,\mathbf{1}%
_{Y\geq z}\right] +\mathbb{E}\left[ z\,\mathbf{1}_{Y\geq z}\right]
\end{equation*}%
(the equality uses the independence of $Y$ and $Z).$ The first term is the
revenue from a subdomain of values of $Y,$ and so is at most its maximal
revenue \textsc{Rev}$(Y)$ by Proposition \ref{p:sub-dom} (iv).\footnote{$\mu
^{z}$ need not satisfy NPT, as $s^{z}$ may take negative values.} As for the
second term, we have 
\begin{equation}
\mathbb{E}\left[ z\,\mathbf{1}_{Y\geq z}\,\right] =z\cdot \mathbb{P}\left[
Y\geq z\right] \,\leq \text{\textsc{Rev}}(Y),  \label{eq:e(min)-le-rev(max)}
\end{equation}%
since posting a price of $z,$ and the buyer buying when $Y\geq z,$
constitutes an IC and IR mechanism for\footnote{%
We are \emph{not} using here the characterization (\ref{eq:one good}) of
optimal one-good mechanisms as posting-price mechanisms, but only the simple
fact that these mechanisms are IC and IR.} $y$. Thus 
\begin{equation*}
\mathbb{E[}s(Y,Z)\,\mathbf{1}_{Y\geq Z}~|~Z=z]\leq 2\text{\textsc{Rev}}(Y)
\end{equation*}%
holds for every value $z$ of $Z;$ taking expectation yields (\ref%
{eq:triangle}), completing the proof.
\end{proof}

\bigskip

In Appendix \ref{s:a-comments} we provide a number of observations arising
from this proof.

\section{The General Decomposition Result\label{s:decomposition}}

We generalize the decomposition of the previous section from two goods to
two \emph{sets} of goods. Let now $Y$ be a $k_{1}$-dimensional nonnegative
random variable, and $Z$ a $k_{2}$-dimensional nonnegative random variable
(with $k_{1},k_{2}\geq 1$). While we assume that the vectors $Y$ and $Z$ are 
\emph{independent}, we allow for arbitrary interdependence among the
coordinates of $Y$, and likewise for the coordinates of $Z$.

The main decomposition result is:

\begin{theorem}
\label{th:decomposition} Let $Y$ and $Z$ be multi-dimensional nonnegative
random variables. If $Y$ and $Z$ are independent then 
\begin{eqnarray}
\text{\textsc{Rev}}(Y,Z) &\leq &\text{\textsc{Rev}}(Y)+\text{\textsc{Rev}}%
(Z)+\text{\textsc{BRev}}(Y)+\text{\textsc{BRev}}(Z)  \label{eq:rev+brev} \\
&\leq &2\,(\text{\textsc{Rev}}(Y)+\text{\textsc{Rev}}(Z)).
\label{eq:rev+rev}
\end{eqnarray}
\end{theorem}

The second inequality (\ref{eq:rev+rev}) follows immediately from the first (%
\ref{eq:rev+brev}), because \textsc{BRev}$\leq $\textsc{Rev}. When $Y$ and $%
Z $ are one-dimensional, both inequalities become (\ref{eq:2+2}) of Theorem %
\ref{th:1/2}.

We start with the basic argument that uses the \textquotedblleft marginal"
mechanism on $y$ generated from a mechanism on $(y,z)$ (as in the previous
section). For a $k$-dimensional random valuation $X=(X_{1},...,X_{k}),$ we
use the notation 
\begin{equation*}
\mathrm{Val}(X):=\mathbb{E}\left[ \sum_{i=1}^{k}X_{i}\right] =\sum_{i=1}^{k}%
\mathbb{E}[X_{i}]
\end{equation*}%
for the expected total sum of values (for one-dimensional $X$ we have 
\textsc{$\mathrm{Val}$}$(X)=\mathbb{E}[X]$).

\begin{lemma}[Marginal Mechanism on Subdomain]
\label{marg-sub} Let $Y$ and $Z$ be multi-dimensional nonnegative random
variables, and let $A\subseteq \mathbb{R}^{k_{1}+k_{2}}$ be a set of values
of $(Y,Z)$. If $Y$ and $Z\,$ are independent then 
\begin{equation*}
\text{\textsc{Rev}}((Y,Z)\,\mathbf{1}_{(Y,Z)\in A})\leq \text{\textsc{Rev}}%
(Y)+\QTR{sc}{\mathrm{Val}}(Z\,\mathbf{1}_{(Y,Z\,)\in A}).
\end{equation*}
\end{lemma}

\begin{proof}
For every $z$ put $A_{z}:=\{y|(y,z)\in A\}$. Take an IC and IR mechanism $%
(q,s)$ for $(y,z)$, and fix some value of $z=(z_{1},\ldots ,z_{k_{2}})$. The
induced mechanism on the $y$ goods is IC and IR, but it also hands out
quantities of the $z$ goods. If we modify it so that instead of allocating $%
z_{j}$ with probability $q_{j}=q_{j}(y,z)$, it reduces the buyer's payment
by the amount of $q_{j}z_{j}$, we are left with an IC and IR mechanism, call
it $(q^{z},s^{z})$, for the $y$ goods. Now $s(y,z)=s^{z}(y)+%
\sum_{j}q_{j}z_{j}\leq s^{z}(y)+\sum_{j}z_{j}$, and so, conditioning on $%
Z\,=z$, 
\begin{eqnarray*}
\mathbb{E}\left[ s(Y,Z)\,\mathbf{1}_{(Y,Z)\in A}~|~Z=z\right] &=&\mathbb{E}%
\left[ s(Y,z)\,\mathbf{1}_{Y\in A_{z}}\right] \\
&\leq &\mathbb{E}\left[ s^{z}(Y)\,\mathbf{1}_{Y\in A_{z}}\right] +\mathbb{E}%
\left[ \left( \sum_{j}z_{j}\right) \mathbf{1}_{Y\in A_{z}}\right]
\end{eqnarray*}%
(the equality in the first line is because $Y$ is independent of $Z$). The
first term in the second line is bounded from above by \textsc{Rev}$(Y)$ by
Proposition \ref{p:sub-dom} (iv), and the second term is $\mathbb{E}\left[
\sum_{j}(Z_{j}\,\mathbf{1}_{(Y,Z)\in A})~|~Z=z\right] $; taking expectation
over the values $z$ of $Z$ completes the proof.
\end{proof}

\bigskip

In the case of two goods, i.e., one-dimensional $Y$ and $Z$, the set of
values $A$ for which we bound \textsc{$\mathrm{Val}$}$(Z\,\mathbf{1}%
_{(Y,Z\,)\in A})$ is the set $A=\{(y,z):y\geq z\}$.

\begin{lemma}[Smaller Value]
\label{smaller} Let $Y$ and $Z$ be one-dimensional nonnegative random
variables. If $Y$ and $Z$ are independent then 
\begin{equation*}
\mathrm{Val}(Z\,\mathbf{1}_{Y\geq Z})\leq \text{\textsc{Rev}}(Y).
\end{equation*}
\end{lemma}

\begin{proof}
For every value $z$ of $Z,$ setting the price for $Y$ at $z$ yields a
revenue of $z\cdot \mathbb{P}[Y\geq z],$ which is therefore at most \textsc{%
Rev}$(Y).$ Thus $\mathrm{Val}(Z\,\mathbf{1}_{Y\geq Z})=\mathbb{E}_{z\sim Z}[%
\mathbb{E}[Z\,\mathbf{1}_{Y\geq Z}~|~Z\,=z]]=\mathbb{E}_{z\sim Z}[z\cdot 
\mathbb{P}[Y\geq z]]\leq \mathbb{E}_{z\sim Z}[\text{\textsc{Rev}}(Y)]=$%
\textsc{Rev}$(Y).$
\end{proof}

\bigskip

In the multi-dimensional case we take $A=\{(y,z):\sum_{i}y_{i}\geq
\sum_{j}z_{j}\}$ (where $y_{i}$ and $z_{j}$ are the coordinates of $y$ and $%
z,$ respectively), and get:

\begin{lemma}[Smaller Value for Multiple Goods]
\label{multi-smaller} Let $Y$ and $Z$ be multi-dimensional nonnegative
random variables. If $Y$ and $Z$ are independent then 
\begin{equation*}
\QTR{sc}{\mathrm{Val}}(Z\,\mathbf{1}_{\sum_{i}Y_{i}\geq
\sum_{j}Z\,_{j}})\leq \text{\textsc{BRev}}(Y).
\end{equation*}
\end{lemma}

\begin{proof}
Apply Lemma \ref{smaller} to the one-dimensional random variables $%
\sum_{i}Y_{i}$ and $\sum_{j}Z_{j}$, and use \textsc{Rev}$(\sum_{i}{Y_{i}})=$%
\textsc{BRev}$(Y)$.
\end{proof}

\bigskip

We can now prove our result.

\begin{proof}[Proof of Theorem \protect\ref{th:decomposition}]
We divide the space as follows: 
\begin{equation*}
\text{\textsc{Rev}}(Y,Z)\leq \text{\textsc{Rev}}\left( (Y,Z)\mathbf{1}%
_{\sum_{i}Y_{i}\geq \sum_{j}Z\,_{j}}\right) +\text{\textsc{Rev}}\left( (Y,Z)%
\mathbf{1}_{\sum_{j}Z\,_{j}\geq \sum_{i}Y_{i}}\right)
\end{equation*}%
(the inequality by NPT; see Proposition \ref{p:sub-dom} (iii)). The first
term is at most 
\begin{equation*}
\text{\textsc{Rev}}(Y)+\text{\textsc{$\mathrm{Val}$}}\left( Z\,\mathbf{1}%
_{\sum_{i}Y_{i}\geq \sum_{j}Z\,_{j}}\right) \leq \text{\textsc{Rev}}(Y)+%
\text{\textsc{BRev}}(Y)
\end{equation*}%
by Lemmas \ref{marg-sub} and \ref{multi-smaller}. The second term is bounded
similarly.
\end{proof}

\bigskip

See Appendix \ref{s:a-comments} for some comments on possible
generalizations of this decomposition approach.

\section{Separate and Bundled Selling\label{s:sep-bun}}

In this section we compare the revenue obtainable from the two simple
mechanisms of selling the goods separately and selling them as one bundle.
These mechanisms are simple as they reduce to one-good mechanisms, for which
the Myerson (1981) characterization (\ref{eq:one good}) applies. The results
below are not only interesting in their own right, but also useful when we
make comparisons to the optimal revenue (Theorems \ref{th:srev-k} and \ref%
{th:brev-k}; cf. Section \ref{s:proofs-thm3-4} below).

One advantage of one-good mechanisms is that revenue is \emph{monotonic}
with respect to valuation: increasing the buyer's values can only increase
the seller's revenue. As natural and appealing as this may sound,
monotonicity does \emph{not} extend to the multiple good case; see Hart and
Reny (2015a).

Formally, for $X$ and $Y$ real random variables, $X$ is \emph{(first-order)
stochastically dominated} by $Y$ if for every real $p$ we have $\mathbb{P}%
\left[ X\geq p\right] \leq \mathbb{P}[Y\geq p]$; essentially,\footnote{\label%
{ft:coupling}One may indeed take $X$ and $Y$ to be defined on the same
probability space $\Omega $ and to satisfy $X\leq Y$ pointwise, i.e., $%
X(\omega )\leq Y(\omega )$ for almost every realization $\omega \in \Omega $
(this is called \textquotedblleft coupling" of $X$ and $Y$). See, e.g.,
Shaked and Shantikumar (2010, Theorem 1.A.1).} what this says is that $Y$
gets higher values than $X$. We have:

\begin{proposition}[Monotonicity for One Good]
\label{p:1-monot}Let $X$ and $Y$ be one-good random valuations. If $X$ is
stochastically dominated by $Y$ then \textsc{Rev}$(X)\leq $\textsc{Rev}$(Y).$
\end{proposition}

\begin{proof}
\textsc{Rev}$(X)=\sup_{p}p\cdot \mathbb{P}[X\geq p]\leq \sup_{p}p\cdot 
\mathbb{P}[Y\geq p]=$\textsc{Rev}$(Y)$ by (\ref{eq:one good}).
\end{proof}

\bigskip

This monotonicity property leads one to consider the highest one-good random
valuation with a given revenue. Normalizing the revenue at $1$, this is the
real random variable $V$ that takes values $V\geq 1$ with probabilities $%
\mathbb{P}\left[ V\geq p\right] =1/p$ for every $p\geq 1$. We refer to a
good with random valuation $V$ as an \emph{equal-revenue} (\textrm{$\mathtt{%
ER}$}) good, and to its distribution, i.e., $F_{V}(p)=1-1/p$ and $%
f_{V}(p)=1/p^{2}$ for $p\geq 1,$ as the equal-revenue ($\mathtt{ER}$)
distribution.\footnote{%
Also known as the Pareto distribution with index $1$ and scale $1;$
interestingly, $V$ is \textrm{$\mathtt{ER}$} if and only if $1/V$ is Uniform
on $(0,1].$} Indeed, the revenue \textsc{Rev}$(V)=1$ of an $\mathtt{ER}$
good is obtained at \emph{any} posted price $p\geq 1$ (recall (\ref{eq:one
good})).\footnote{%
Moreover, an IC and IR mechanism $\mu =(q,s)$ is optimal for $V$ if and only
if it does not sell the good for values below $1,$ i.e., $q(x)=s(x)=0$ for
all $x<1,$ and $\sup_{x\geq 1}q(x)=\lim_{x\rightarrow \infty }q(x)=1$ (but
is otherwise arbitrary for $x\geq 1).$ Also, the \textrm{$\mathtt{ER}$}
distribution is the only distribution (up to rescaling) for which the
\textquotedblleft virtual valuation\textquotedblright\ of Myerson (1981)
(used, for instance, when there are multiple buyers) vanishes everywhere in
the support: $x-(1-F(x))/f(x)=0$ for all $x\geq 1.$} Note that while the
revenue of $V$ is finite, its expected value is infinite: $\mathbb{E}\left[ V%
\right] =\int_{1}^{\infty }p\cdot (1/p^{2})\,\mathrm{d}p=\infty .$

The result (\ref{eq:one good}) for one good may now be restated as follows: 
\textsc{Rev}$(X)\leq 1$ if and only if $X$ is stochastically dominated by an 
$\mathtt{ER}$ good $V$ (indeed, \textsc{Rev}$(X)\leq 1$ if and only if $%
1-F_{X}(p)\leq 1/p=1-F_{V}(p)$ for all $p\geq 1).$ That is, the revenue from
a one-good random valuation $X$ is at most $1$ if and only if one can
increase the values of $X$ and obtain a new random valuation $V$ that is 
\textrm{$\mathtt{ER}$}-distributed.

The following proposition collects the above observation together with a
number of useful results on the \textrm{$\mathtt{ER}$} distribution; the
proofs are relegated to Appendix \ref{s:ER}. From now on we use the constant 
$w\approx 0.278$ to denote the solution of the equation\footnote{%
Thus $we^{w}=1/e,$ and so $w=W(1/e)$ where $W$ is the so-called
\textquotedblleft Lambert-$W$\textquotedblright\ function.} 
\begin{equation*}
we^{w+1}=1.
\end{equation*}

\begin{proposition}
\label{p:ER}

\begin{description}
\item[(i)] Let $X$ be a one-good random valuation, and let $r\geq 0.$ Then 
\textsc{Rev}$(X)\leq r$ if and only if $X$ is stochastically dominated by $%
rV $ where $V$ is an $\mathtt{ER}$ valuation.

\item[(ii)] Let $V_{1},V_{2},...,V_{k}$ be i.i.d.-$\mathtt{ER}$, let $%
r_{1},r_{2},...r_{k}\geq 0,$ and put $\bar{r}:=(1/k)\sum_{i=1}^{k}r_{i}$ for
the average of the $r_{i}.$ Then $\sum_{i=1}^{k}r_{i}V_{i}$ is
stochastically dominated by $\sum_{i=1}^{k}\bar{r}V_{i}.$

\item[(iii)] Let $V_{1}$ and $V_{2}$ be i.i.d.-$\mathtt{ER}$. Then 
\begin{equation*}
\text{\textsc{BRev}}(V_{1},V_{2})=2(w+1)\approx 2.56.
\end{equation*}

\item[(iv)] There exist constants $c_{1}>0$ and $c_{2}<\infty $ such that
for all $k\geq 2$ and $V_{1},V_{2},...,V_{k}$ i.i.d.-$\mathtt{ER,}$ 
\begin{equation*}
c_{1}k\log k\leq \text{\textsc{BRev}}\left( V_{1},V_{2},...,V_{k}\right)
\leq c_{2}k\log k.
\end{equation*}
\end{description}
\end{proposition}

\bigskip

\noindent \textbf{Remarks. }\emph{(a) }We will see below (Corollary \ref%
{c:rev(er,er)}) that bundling is in fact optimal for two i.i.d.-$\mathtt{ER}$
goods, and so (iii) will become \textsc{Rev}$(V_{1},V_{2})=$\textsc{BRev}$%
(V_{1},V_{2})$ $=$\textsc{Rev}$(V_{1}+V_{2})=2(w+1).$

\emph{(b)} The fact that the revenue is \emph{not} monotonic for multiple
goods (Hart and Reny 2015a) foils the following natural attempt to estimate 
\textsc{GFOR} for separate selling. For concreteness, consider two i.i.d.
goods $X_{1}$ and $X_{2},$ without loss of generality normalized so that 
\textsc{Rev}$(X_{1})=\text{\textsc{Rev}}(X_{2})=1.$ Let $V_{1},V_{2}$ be
i.i.d.-$\mathrm{\mathtt{ER.}}$ Then each $X_{i}$ is stochastically dominated
by $V_{i}$, and so $X=(X_{1},X_{2})$ is stochastically dominated by $%
V=(V_{1},V_{2})\mathrm{.}$ However, we \emph{cannot} deduce from this that 
\textsc{Rev}$(X)\leq \text{\textsc{Rev}}(V)=2(w+1)$ (see Remark (a)
above)---which would have given a better bound of $1/(w+1)\approx 0.78$, and
with a much simpler proof, for \textsc{GFOR}(\textsc{separate}) in this case
(cf. Theorem \ref{th:73} and its proof in Appendix \ref{s:proof-iid}).

\bigskip

Using $\mathtt{ER}$ goods allows us to compare the separate selling revenue
to the bundling revenue.

\begin{proposition}
\label{p:sep>bun}

\begin{description}
\item[(i)] For any two independent goods $X_{1},X_{2},$ 
\begin{equation*}
\text{\textsc{SRev}}(X_{1},X_{2})\geq \frac{1}{w+1}\text{\textsc{BRev}}%
(X_{1},X_{2})\approx 0.78\cdot \text{\textsc{BRev}}(X_{1},X_{2}).
\end{equation*}

\item[(ii)] There exists a constant $c>0$ such that for any $k\geq 2$ and
any $k$ independent goods $X_{1},X_{2},...,X_{k},$%
\begin{equation*}
\text{\textsc{SRev}}(X_{1},X_{2},...,X_{k})\geq \frac{c}{\log k}\text{%
\textsc{BRev}}(X_{1},X_{2},...,X_{k}).
\end{equation*}
\end{description}
\end{proposition}

\begin{proof}
Put $r_{i}:=$\textsc{Rev}$(X_{i})$ and $\bar{r}:=(1/k)\sum_{i}r_{i}=(1/k)$%
\textsc{SRev}$(X),$ and let $V_{1},...,V_{k}$ be $k$ i.i.d.-$\mathtt{ER}$
goods. Using Proposition \ref{p:ER} (i) and (ii): each $X_{i}$ is
stochastically dominated by $r_{i}V_{i}$, hence $\sum_{i}X_{i}$ is
stochastically dominated by\footnote{\label{ft:stoch-dom-convolution}We use
here the following fact: if $X_{i}$ is stochastically dominated by $Y_{i}$
for every $i,$ then $X_{1}+\cdots +X_{k}$ is stochastically dominated by $%
Y_{1}+\cdots +Y_{k}$ (this is immediate when all the random variables are
defined on the same probability space and $X_{i}\leq Y_{i}$ pointwise for
every $i$---cf. the coupling in footnote \ref{ft:coupling}---because then $%
\sum X_{i}\leq \sum Y_{i}$); see, e.g., Shaked and Shantikumar (2010,
Theorem 1.A.3.(b)).} $\sum_{i}r_{i}V_{i},$ which is in turn dominated by $%
\bar{r}\sum_{i}V_{i}.$ Therefore 
\begin{eqnarray*}
\text{\textsc{BRev}}(X_{1},...,X_{k}) &=&\text{\textsc{Rev}}\left(
\sum_{i=1}^{k}X_{i}\right) \leq \text{\textsc{Rev}}\left( \bar{r}%
\sum_{i=1}^{k}V_{i}\right) \\
&=&\bar{r}~\text{\textsc{Rev}}\left( \sum_{i=1}^{k}V_{i}\right) =\bar{r}~%
\text{\textsc{BRev}}(V_{1},...,V_{k})
\end{eqnarray*}%
(the inequality is by monotonicity for one good, Proposition \ref{p:1-monot}%
), and then the two results follow from Proposition \ref{p:ER} (iii) and
(iv), respectively.
\end{proof}

\bigskip

Taking the goods $X_{i}$ to be $\mathtt{ER}$ goods shows that $1/(w+1)$ and $%
c/\log k$ above are both tight (cf. Proposition \ref{p:ER} (iii) and (iv)).

We conclude with comparisons in the other direction: the bundled revenue as
a fraction of the separate revenue; see Appendix \ref{s:ap-sep-bun} for
additional results.

\begin{proposition}
\label{p:bun>sep}

\begin{description}
\item[(i)] For any $k\geq 1$ and any $k$ independent goods $%
X_{1},X_{2},...,X_{k},$ 
\begin{equation*}
\text{\textsc{BRev}}(X_{1},X_{2},...,X_{k})\geq \frac{1}{k}\text{\textsc{SRev%
}}(X_{1},X_{2},...,X_{k}).
\end{equation*}

\item[(ii)] For any $k\geq 1$ and any $k$ i.i.d. goods $%
X_{1},X_{2},...,X_{k},$ 
\begin{equation*}
\text{\textsc{BRev}}(X_{1},X_{2},...,X_{k})\geq \frac{1}{4}\text{\textsc{SRev%
}}(X_{1},X_{2},...,X_{k}).
\end{equation*}
\end{description}
\end{proposition}

\begin{proof}
\textbf{(i)} For every $i$ we have $X_{i}\leq \sum_{j}X_{j}$ and so \textsc{%
Rev}$(X_{i})\leq $\textsc{Rev}$(\sum_{j}X_{j})=$ \textsc{BRev}$%
(X_{1},...,X_{k})$; summing over $j$ yields $\sum_{j}$\textsc{Rev}$%
(X_{j})\leq k~$\textsc{BRev}$(X_{1},...,X_{k})$.

\textbf{(ii)} Let $p$ be an optimal one-good price for each $X_{i}$ and put $%
\alpha :=\mathbb{P}[X_{i}\geq p]$; thus \textsc{Rev}$(X_{i})=p\alpha $. We
separate between two cases. If $k\alpha \leq 1$ then consider setting the
bundle price at $p;$ the probability that the buyer will buy is%
\begin{eqnarray*}
\mathbb{P}\left[ \sum_{i}X_{i}\geq p\right] &\geq &\mathbb{P}\left[
X_{i}\geq p\text{ for some }i\right] =\mathbb{P}\left[ \bigcup_{i}[X_{i}\geq
p]\right] \\
&\geq &\sum_{i}\mathbb{P}\left[ X_{i}\geq p\right] -\sum_{i<j}\mathbb{P}%
\left[ X_{i}\geq p,X_{j}\geq p\right] \\
&=&k\alpha -{\binom{k}{2}}\alpha ^{2}\geq \frac{1}{2}k\alpha ,
\end{eqnarray*}%
and so the revenue will be at least $pk\alpha /2\geq k~$\textsc{Rev}$%
(X_{i})/2$. If $k\alpha \geq 1$ then consider setting the bundle price at $%
p\lfloor k\alpha \rfloor $. Since the median in the $\mathrm{Binomial}%
(k,\alpha )$ distribution is at least $\lfloor k\alpha \rfloor $, the
probability that the buyer will buy is at least $1/2$, and so the revenue
will be at least $p\lfloor k\alpha \rfloor /2\geq pk\alpha /4=k~$\textsc{Rev}%
$(X_{i})/4$.
\end{proof}

\bigskip

While the constant $1/k$ in (i) is tight, the $1/4$ in (ii) is not (we have
not attempted to optimize it); see Example \ref{ex:1/k} in Appendix \ref%
{s:ap-sep-bun} and Example \ref{ex:57} in Appendix \ref{s:k-iid}.

\section{$k$ Independent Goods\label{s:proofs-thm3-4}}

We now prove the two main results on $k\geq 2$ goods, Theorems \ref%
{th:srev-k} and \ref{th:brev-k} stated in the Introduction, using our
general decomposition result of Theorem \ref{th:decomposition}.

We start with separate selling. Viewing $2k$ goods as two sets of $k$ goods
each and using (\ref{eq:rev+rev}) one can easily get by induction that $%
\sum_{i=1}^{k}$\textsc{Rev}$(X_{i})\geq (1/k)$\textsc{Rev}$(X_{1},...,X_{k})$%
, as follows: 
\begin{eqnarray*}
\text{\textsc{Rev}}(X_{1},...,X_{2k}) &\leq &2(\text{\textsc{Rev}}%
(X_{1},...,X_{k})+\text{\textsc{Rev}}(X_{k+1},...,X_{2k})) \\
&\leq &2\left( k\sum_{i=1}^{k}\text{\textsc{Rev}}(X_{i})+k\sum_{i=k+1}^{2k}%
\text{\textsc{Rev}}(X_{i})\right) \\
&=&2k\sum_{i=1}^{2k}\text{\textsc{Rev}}(X_{i}).
\end{eqnarray*}%
However, using the stronger inequality (\ref{eq:rev+brev}), together with
the relations we have shown in the previous section between the bundling and
the separate revenues, gives us the better bound of $c/\log ^{2}k$ (instead
of $1/k$) of Theorem \ref{th:srev-k}.

\bigskip

\begin{proof}[Proof of Theorem \protect\ref{th:srev-k}]
We will first prove by induction that \textsc{Rev}$(X_{1},...,X_{k})\leq
(1/c^{\prime })\log _{2}^{2}k\sum_{i=1}^{k}$\textsc{Rev}$(X_{i})$ for every $%
k\geq 2$ that is a power of $2,$ where $c^{\prime }:=\min \{c,1/2\}>0$ with $%
c>0$ given by Proposition \ref{p:sep>bun} (ii). This inequality holds for $%
k=2$ by Theorem \ref{th:1/2} (since $c^{\prime }\leq 1/2).$ For $k\geq 4$ we
apply Theorem \ref{th:decomposition} to $Y=(X_{1},...,X_{k})$ and $%
Z=(X_{k+1},...,X_{2k}),$ to get 
\begin{eqnarray}
\text{\textsc{Rev}}(X_{1},...,X_{2k}) &\leq &\text{\textsc{Rev}}%
(X_{1},...,X_{k})+\text{\textsc{Rev}}(X_{k+1},...,X_{2k})  \notag \\
&&+\,\text{\textsc{BRev}}(X_{1},...,X_{k})+\text{\textsc{BRev}}%
(X_{k+1},...,X_{2k}).  \label{eq:k/2}
\end{eqnarray}%
First, using Proposition \ref{p:sep>bun} (ii) (and $c^{\prime }\leq c)$ on
each of the \textsc{BRev} terms shows that their sum is bounded by $%
(1/c^{\prime })\log _{2}k\sum_{i=1}^{2k}$\textsc{Rev}$(X_{i}).$ Second,
using the induction hypothesis on each of the \textsc{Rev} terms shows that
their sum is bounded by $(1/c^{\prime })\log _{2}^{2}k\sum_{i=1}^{2k}$%
\textsc{Rev}$(X_{i})$. Now $\log _{2}k+\log _{2}^{2}k\leq \log _{2}^{2}(2k)$%
, and so adding the two bounds gives the result.

Next, when $2^{m-1}<k<2^{m}$ we can \textquotedblleft pad" to $2^{m}$ goods
by adding goods that have value identically zero, and so do not contribute
anything to the revenue; this at most doubles $k$.
\end{proof}

\bigskip

In Appendix \ref{s:sep-bun} we show that bundling may, by contrast, extract
only a $1/k$ fraction of the optimal revenue. However, bundling does much
better for identically distributed goods, and in fact we have a tighter
result, Theorem \ref{th:brev-k}, with $\log k$ instead of $k$.

\bigskip

\begin{proof}[Proof of Theorem \protect\ref{th:brev-k}]
Let $X_{i}$ be i.i.d., and put $R_{k}:=$\textsc{Rev}$(X_{1},...,X_{k})$ and $%
B_{k}:=$\textsc{BRev}$(X_{1},...,X_{k}).$ We want to show that there is a
finite $c>0$ such that $R_{k}\leq c\log k~B_{k}$ for all $k\geq 2.$ If $%
k\geq 2$ is a power of $2$ we apply Theorem \ref{th:decomposition}
inductively to obtain $R_{k}\leq
2B_{k/2}+4B_{k/4}+...+(k/2)B_{2}+kB_{1}+kR_{1}.$ Each of the $\log _{2}k+1$
terms in this sum is of the form $(k/\ell )B_{\ell }=(k/\ell )$\textsc{Rev}$%
(X_{1}+...+X_{\ell }),$ and is thus bounded from above by $4B_{k}$ (apply
Proposition \ref{p:bun>sep} (ii) to $k/\ell $ i.i.d. random variables each
distributed as $X_{1}+...+X_{\ell }).$ Altogether we have $R_{k}\leq 4(\log
_{2}k+1)B_{k}.$

When $2^{m-1}<k<2^{m}$ we have $R_{k}\leq R_{2^{m}}$ and $B_{k}\geq
B_{2^{m-1}}$ (adding goods can only increase the revenue: take the optimal
mechanism for the original set of goods and extend it so that it ignores the
additional goods), and $B_{2^{m}}\leq 2(w+1)B_{2^{m-1}}\leq 2(w+1)B_{k}$
(apply Proposition \ref{p:sep>bun} (i) to the two i.i.d. random variables $%
X_{1}+...+X_{2^{m-1}}$ and $X_{2^{m-1}+1}+...+X_{2^{m}}$), which together
with the above inequality for $2^{m}$ goods yields $R_{k}\leq 8(w+1)(\log
_{2}k+2)B_{k}$.
\end{proof}

\section{Additional Results\label{s:additional}}

\subsection{Upper Bound on GFOR for Two Goods\label{sus:gfor-le-0.78}}

Our results for two goods give lower bounds on \textsc{GFOR} ($50\%$ and $%
73\%$ for selling separately two independent goods and two i.i.d. goods,
respectively). Now what about upper bounds? That is, how high can \textsc{%
GFOR(separate)} actually be? The best estimate we have is that it cannot
exceed approximately $78\%$: there exist two i.i.d. goods where selling
separately yields only that fraction of the optimal revenue (while \textsc{%
GFOR} may well be lower for independent goods than for the more restricted
i.i.d. goods, we have not found a better example---i.e., with a lower
fraction---in the former class).

\begin{proposition}
\label{p:gfor-le}In the case of two independent goods, selling separately
cannot guarantee more than $78\%$ of the optimal revenue; i.e.,%
\begin{eqnarray*}
&&\text{\textsc{GFOR}}(\text{\textsc{separate}};~2\text{ independent goods})
\\
&\leq &\text{\textsc{GFOR}}(\text{\textsc{separate}};~2\text{ i.i.d. goods}%
)\leq \frac{1}{w+1}\approx 0.78.
\end{eqnarray*}
\end{proposition}

\begin{proof}
Let $V=(V_{1},V_{2})$ with $V_{1}$ and $V_{2}$ two i.i.d.-\textrm{$\mathtt{ER%
}$} goods. The revenue from selling separately is \textsc{SRev}$(V)=\,$%
\textsc{Rev}$(V_{1})+$\textsc{Rev}$(V_{2})=2,$ whereas, as shown in the next
section (Corollary \ref{c:rev(er,er)}), the optimal revenue is \textsc{Rev}$%
(V)=2(w+1)$ (obtained by bundling).
\end{proof}

\subsection{When Bundling Is Optimal\label{sus:indep-bundled}}

Interestingly, we have identified a class of two-good i.i.d. distributions
for which bundling is optimal.

\begin{theorem}
\label{th:brev-opt} Let $F$ be a continuous one-good distribution with
values in $[a,\infty )$ for some $a>0,$ and density function $f$ that is
differentiable and satisfies%
\begin{equation}
xf^{\prime }(x)+\frac{3}{2}f(x)\leq 0  \label{eq:3/2}
\end{equation}%
for every $x>a.$ Then bundling is optimal for two i.i.d.-$F$ goods $%
X_{1},X_{2}$: 
\begin{equation*}
\text{\textsc{Rev}}(X_{1},X_{2})=\text{\textsc{BRev}}(X_{1},X_{2})=\text{%
\textsc{Rev}}(X_{1}+X_{2}).
\end{equation*}
\end{theorem}

Theorem \ref{th:brev-opt} is proved in Appendix \ref{s:brev=rev}. Condition (%
\ref{eq:3/2}) is equivalent to $\left( x^{3/2}f(x)\right) ^{\prime }\leq 0$,
i.e., $x^{3/2}f(x)$ is nonincreasing in $x$ (the support of $f$ is thus
either some finite interval $[a,b]$ or the half-line $[a,\infty )$). When $%
f(x)=cx^{-\gamma },$ (\ref{eq:3/2}) holds whenever $\gamma \geq 3/2$. In
particular, the \textrm{$\mathtt{ER}$} distribution (where $\gamma =2)$
satisfies (\ref{eq:3/2}), and so does any general Pareto distribution with
index $\alpha \geq 1/2$. Together with Proposition \ref{p:ER} (iii) we thus
get:

\begin{corollary}
\label{c:rev(er,er)} Let $V_{1},V_{2}$ be two i.i.d.-\textrm{$\mathtt{ER}$}
goods. Then 
\begin{equation*}
\text{\textsc{Rev}}(V_{1},V_{2})=\text{\textsc{BRev}}(V_{1},V_{2})=2(w+1)%
\approx 2.56.
\end{equation*}
\end{corollary}

\subsection{Multiple Buyers\label{sus:indep-n-buyers}}

Up to now we have been dealing with a single buyer, but our result for two
independent goods turns out to hold also when there are multiple buyers.
Unlike the simple decision-theoretic problem facing a single buyer, we now
have a multi-person game among the buyers. Two main notions of equilibrium
are considered: \emph{dominant strategy} equilibrium and \emph{Bayesian Nash}
equilibrium (corresponding to \textquotedblleft ex-post\textquotedblright\
and \textquotedblleft interim" implementations, respectively); see Appendix %
\ref{s:n-buyers-proofs} for details. Our result holds for both concepts.

\begin{theorem}
\label{th:gfor-n}In the case of $n$ independent buyers and two goods, if the
random valuations of the two goods are independent, then selling each good
separately using its optimal one-good mechanism guarantees at least $50\%$
of the optimal revenue: 
\begin{equation*}
\text{\textsc{GFOR}}(\text{\textsc{separate}})\geq \frac{1}{2};
\end{equation*}%
this holds when the optimal revenue is taken throughout\footnote{%
I.e., for the two goods, as well as for each good separately.} with respect
to either dominant strategy implementation or Bayesian Nash implementation.
\end{theorem}

That is, let the one-dimensional random variable $X_{i}^{j}\geq 0$ denote
the value of good $i$ to buyer $j,$ for $i=1,2$ and $j=1,...,n.$ Write $%
X^{j}=(X_{1}^{j},X_{2}^{j})\in \mathbb{R}_{+}^{2}$ for the random valuation
vector of buyer $j$ for both goods, and $X_{i}=(X_{i}^{j})_{j=1,...,n}\in 
\mathbb{R}_{+}^{n}$ for the vector of values of all buyers for good $i.$ 
\emph{Independent buyers}\ means that the random vectors $%
X^{1},X^{2},...,X^{n}$ are independent; \emph{independent goods}\ means that
the random vectors $X_{1}$ and $X_{2}$ are independent.\footnote{%
Independent buyers together with independent goods means that the $2n$
random variables $X_{i}^{j}$ are all independent.} Theorem \ref{th:gfor-n}
is proved in Appendix \ref{s:n-buyers-proofs}, which also contains the
precise notations and statements; the proof is again a generalization of the
proof of Theorem \ref{th:1/2} for one buyer (Section \ref{s:independent}).

\bigskip

\noindent \textbf{Remarks.} \emph{(a) Dependent buyers and dominant strategy
implementation. }In the dominant strategy case, our proof does not use the
independence between the buyers' random valuations; thus \textsc{GFOR}$($%
\textsc{separate\thinspace }$)\geq 1/2$ holds under dominant strategy
implementation for two independent goods and any number of buyers, whether
independent or not; see Theorem \ref{th:n-DS} in Appendix \ref%
{s:n-buyers-proofs}.

\emph{(b) Dependent buyers and Bayesian Nash implementation.} In the
Bayesian Nash case our proof does not extend when the buyers are not
independent (see Appendix \ref{s:n-buyers-proofs}). However, for this case Cr%
\'{e}mer and McLean (1988) show that, under a certain general
\textquotedblleft correlation-between-buyers" condition, the seller can
extract all the surplus from any single good: \textsc{Rev}$(X_{i})=\mathbb{E}%
\left[ \max_{1\leq j\leq n}X_{i}^{j}\right] .$ Since the most that the
seller can extract from the two goods is, by IR,\footnote{%
Indeed (see Appendix \ref{s:n-buyers-proofs} for notations), $b^{j}(x)\geq 0$
implies that $s^{j}(x)\leq q^{j}(x)\cdot x^{j},$ and so $\sum_{j}s^{j}(x)%
\leq \sum_{j}q^{j}(x)\cdot x^{j}=\sum_{i}\sum_{j}q_{i}^{j}(x)\,x_{i}^{j}\leq
\sum_{i}\max_{j}x_{i}^{j}.$} $\mathbb{E}\left[ \max_{j}X_{1}^{j}\right] +%
\mathbb{E}\left[ \max_{j}X_{2}^{j}\right] $, it follows that in this case 
\textsc{Rev}$(X_{1},X_{2})=\,$\textsc{Rev}$(X_{1})+$\textsc{Rev}$(X_{2}),$
and so \textsc{GFOR}$($\textsc{separate}$)=1.$ We do not know what \textsc{%
GFOR(separate)} is when the buyers are neither independent nor satisfy the Cr%
\'{e}mer--McLean condition.

\section{Open Problems\label{s:discussion}}

Many interesting problems remain open. As attested by the long time that has
passed since Myerson's (1981) work in the one-good case, characterizing the
optimal mechanisms in the multiple-goods case---even when there are just two
goods---is an extremely difficult problem.\footnote{%
Recall footnote \ref{ftn:conceptual-complexity} on \emph{conceptual
complexity}.} While the general problem appears very complex, one may well
be able to obtain results for certain useful classes of random valuations
and mechanisms. Following are some specific questions that arise from our
study:

\begin{enumerate}
\item Characterize distributions where separate selling is optimal (cf.
Theorem \ref{th:brev-opt} for bundling).

\item Provide bounds for \textsc{GFOR(deterministic)}, the fraction of
optimal revenue that is guaranteed by mechanisms that do not use
randomizations. In addition, characterize distributions where deterministic
mechanisms are optimal.\footnote{%
The recent work of Babaioff, Immorlica, Lucier, and Weinberg (2014) implies
in particular that \textsc{GFOR(deterministic)} is bounded from below by a
constant that is independent of the number of goods.}

\item Tighten the bounds on \textsc{GFOR(separate)}. While the gap in the
i.i.d. case ($73\%$ vs. $78\%$) is quite small, we do not know what the
right value is; also, is \textsc{GFOR} in the independent case in fact lower
than in the i.i.d. case? Is $50\%$ the right bound?

\item Evaluate \textsc{GFOR(separate) }for Bayesian Nash implementation when
there are multiple buyers that are neither independent nor satisfy the Cr%
\'{e}mer--McLean condition (see Remark (b) in Section \ref%
{sus:indep-n-buyers}).

\item Find simple mechanisms different from separate selling that can
guarantee a larger fraction of the optimal revenue.

\item Study the case of two or more goods that are not necessarily
independent (see Hart and Nisan 2013).\footnote{%
Where it is shown that no simple class of mechanisms can guarantee any
positive fraction of the optimal revenue; i.e., \textsc{GFOR}$=0.$}

\item Obtain useful ways to quantify the complexity (vs. simplicity) of
mechanisms, and analyze the tradeoffs between complexity and revenue (see
Hart and Nisan 2013 for such an approach: \textquotedblleft menu
complexity").
\end{enumerate}

\appendix

\section{Appendix}

\subsection{Proof for Two I.I.D. Goods\label{s:proof-iid}}

In this appendix we prove Theorem \ref{th:73}, stated in the Introduction,
which says that selling two i.i.d. goods separately yields at least $e/(e+1)$
of the optimal revenue. The proof follows a line of argument similar to that
of the proof of Theorem \ref{th:1/2} in Section \ref{s:independent}, but is
more intricate as we make all the estimates much tighter.

\bigskip

\begin{proof}[Proof of Theorem \protect\ref{th:73}]
Let $X=(Y,Z),$ where $Y$ and $Z$ are i.i.d. nonnegative one-dimensional
random variables, and let $r:=\,$\textsc{Rev}$(Y)=\,$\textsc{Rev}$%
(Z)=\sup_{t\geq 0}t\cdot G(t)$ be the revenue from each good separately,
where $G(t):=\mathbb{P}\left[ Y\geq t\right] $. We want to prove that 
\begin{equation*}
\text{\textsc{Rev}}(Y,Z)\leq \frac{e+1}{e}(\text{\textsc{Rev}}(Y)+\text{%
\textsc{Rev}}(Z))=2\left( 1+\frac{1}{e}\right) r.
\end{equation*}

Take a two-good IC and IR mechanism $\mu =(q,s)$ with buyer payoff function $%
b$ (i.e., $b(x)=q(x)\cdot x-s(x)$ for all $x).$ Without loss of generality
assume that it also satisfies NPT, i.e., $s(x)\geq 0$ for all $x,$ and $%
b(0,0)=s(0,0)=0$ (recall Proposition \ref{p:sub-dom} (iii))$.$ Since $X$ is
symmetric we will also assume that $\mu $ is symmetric, i.e., $%
q_{1}(y,z)=q_{2}(z,y)$ and $s(y,z)=s(z,y)$---and thus $b(y,z)=b(z,y)$---for
all $y,z\geq 0.$ Indeed, $\mu $ can be replaced by its \textquotedblleft
symmetrization\textquotedblright\ $\bar{\mu}=(\bar{q},\bar{s})$ given by $%
\bar{q}_{1}(y,z)=\bar{q}_{2}(z,y):=(q_{1}(y,z)+q_{2}(z,y))/2$ and $\bar{s}%
(y,z)=\bar{s}(z,y):=(s(y,z)+s(z,y))/2$ for all $y,z\geq 0,$ which also
satisfies IC, IR, and NPT, and yields the same revenue, $\mathbb{E}\left[ 
\bar{s}(Y,Z)\right] =\mathbb{E}\left[ s(Y,Z)\right] =\mathbb{E}\left[ s(Z,Y)%
\right] $ (because $Y,Z$ are i.i.d.).

For every $t\geq 0$ put $\Phi (t):=b(t,t)/2$ and $\varphi
(t):=q_{1}(t,t)=q_{2}(t,t);$ Proposition \ref{p:IC =b-convex} implies that $%
\Phi $ is a convex function, $\varphi (t)=\Phi ^{\prime }(t)$ almost
everywhere, and $\Phi (u)=\int_{0}^{u}\varphi (t)\,$\textrm{d}$t$ (formally,
use Corollary 24.2.1 in Rockafellar 1970 and $\Phi (0)=b(0,0)=0)$.

Consider first the region $Y\geq Z.$ For each fixed $z\geq 0$ such that $%
\mathbb{P}\left[ Y\geq z\right] >0$ define a mechanism $\mu
^{z}=(q^{z},s^{z})$ for the first good by $q^{z}(y):=q_{1}(y,z)$ and $%
s^{z}(y):=s(y,z)-q_{2}(y,z)\cdot z$ for every $y\geq 0;$ the buyer's payoff
remains the same: $b^{z}(y)=b(y,z)$. The mechanism $\mu ^{z}$ is IC and IR
for $y$, since $\mu $ is IC and IR for $(y,z)$. Let $Y^{z}$ denote the
random variable $Y$ conditional on the event $Y\geq z,$ and consider the
revenue $R(\mu ^{z};Y^{z})=\mathbb{E}\left[ s^{z}(Y^{z})\right] =\mathbb{E}%
\left[ s^{z}(Y)|Y\geq z\right] $ of $\mu ^{z}$ from $Y^{z}$. We have $%
Y^{z}\geq z,$ $q^{z}(z)=\varphi (z),$ and $s^{z}(z)=s(z,z)-q_{2}(z,z)\cdot
z=q_{1}(z,z)\cdot z-b(z,z)=z\varphi (z)-2\Phi (z),$ and so applying Lemma %
\ref{l:constrained-R1} below to $Y^{z}$ yields%
\begin{equation}
\mathbb{E}\left[ s^{z}(Y)|Y\geq z\right] \leq (1-\varphi (z))\text{\textsc{%
Rev}}(Y^{z})+z\varphi (z)-2\Phi (z).  \label{eq:X-tilda}
\end{equation}%
Since $\mathbb{P}\left[ Y^{z}\geq t\right] =\mathbb{P}\left[ Y\geq t\right] /%
\mathbb{P}\left[ Y\geq z\right] =G(t)/\mathbb{P}\left[ Y\geq z\right] $ for
all $t\geq z,$ we get from (\ref{eq:one good}) that%
\begin{equation*}
\text{\textsc{Rev}}(Y^{z})=\sup_{t\geq 0}t\cdot \mathbb{P}\left[ Y^{z}\geq t%
\right] =\sup_{t\geq z}t\cdot \frac{G(t)}{\mathbb{P}\left[ Y\geq z\right] }%
\leq \frac{\sup_{t\geq 0}t\cdot G(t)}{\mathbb{P}\left[ Y\geq z\right] }=%
\frac{r}{\mathbb{P}\left[ Y\geq z\right] }
\end{equation*}%
(recall that $r=$\textsc{Rev}$(Y)).$ Substitute this in (\ref{eq:X-tilda}),
and multiply it by $\mathbb{P}\left[ Y\geq z\right] ,$ to get%
\begin{equation*}
\mathbb{E}\left[ s^{z}(Y)\mathbf{1}_{Y\geq z}\right] \leq r(1-\varphi
(z))+(z\varphi (z)-2\Phi (z))\mathbb{P}\left[ Y\geq z\right]
\end{equation*}%
for all $z\geq 0$ (trivially including those where $\mathbb{P}\left[ Y\geq z%
\right] =0).$ Taking expectation over the values $z$ of $Z$:%
\begin{equation}
\mathbb{E}\left[ s^{Z}(Y)\mathbf{1}_{Y\geq Z}\right] \leq r(1-\,\mathbb{E}%
\left[ \varphi (Z)\right] )+\mathbb{E}\left[ (Z\varphi (Z)-2\Phi (Z))\mathbf{%
1}_{Y\geq Z}\right] .  \label{eq:e[s-hat]}
\end{equation}%
Now $s(y,z)=s^{z}(y)+q_{2}(y,z)\,z\leq
s^{z}(y)+q_{2}(y,y)\,z=s^{z}(y)+z\varphi (y)$ (use $z\geq 0$ and the
monotonicity of $q_{2}(y,z)=b_{z}(y,z)$ in $z,$ again from the convexity of $%
b),$ which together with (\ref{eq:e[s-hat]}) yields%
\begin{eqnarray*}
\mathbb{E}\left[ s(Y,Z)\mathbf{1}_{Y\geq Z}\right] &\leq &\mathbb{E}\left[
s^{Z}(Y)\mathbf{1}_{Y\geq Z}\right] +\mathbb{E}\left[ Z\varphi (Y)\mathbf{1}%
_{Y\geq Z}\right] \\
&\leq &r(1-\,\mathbb{E}\left[ \varphi (Z)\right] )+\mathbb{E}\left[
(Z\varphi (Y)+Z\varphi (Z)-2\Phi (Z))\mathbf{1}_{Y\geq Z}\right] \\
&=&r(1-\,\mathbb{E}\left[ \varphi (Z)\right] )+\mathbb{E}\left[ (\Lambda
\varphi (Y)+\Lambda \varphi (Z)-2\Phi (\Lambda ))\mathbf{1}_{Y\geq Z}\right]
\end{eqnarray*}%
where we put $\Lambda :=\min \{Y,Z\}.$

Consider next the region $Z>Y.$ Interchanging $Y$ and $Z$ and using $Z>y$
instead of $Z\geq y$ throughout gives%
\begin{equation*}
\mathbb{E}\left[ s(Y,Z)\mathbf{1}_{Z>Y}\right] \leq r(1-\,\mathbb{E}\left[
\varphi (Y)\right] )+\mathbb{E}\left[ (\Lambda \varphi (Z)+\Lambda \varphi
(Y)-2\Phi (\Lambda ))\mathbf{1}_{Z>Y}\right] .
\end{equation*}%
Adding the last two inequalities yields%
\begin{eqnarray*}
\mathbb{E}\left[ s(Y,Z)\right] &\leq &r(2-\mathbb{E}\left[ \varphi (Y)\right]
-\mathbb{E}\left[ \varphi (Z)\right] ) \\
&&+\mathbb{E}\left[ \Lambda \varphi (Y)+\Lambda \varphi (Z)-2\Phi (\Lambda )%
\right] .
\end{eqnarray*}%
Because $Y$ and $Z$ are i.i.d. we have $\mathbb{E}\left[ \varphi (Y)\right] =%
\mathbb{E}\left[ \varphi (Z)\right] $ and $\mathbb{E}\left[ \Lambda \varphi
(Y)\right] =\mathbb{E}\left[ \Lambda \varphi (Z)\right] ,$ and so%
\begin{equation}
\mathbb{E}\left[ s(Y,Z)\right] \leq 2r-2r\mathbb{E}\left[ \varphi (Y)\right]
+2\mathbb{E}\left[ W\right]  \label{eq:upper-bound}
\end{equation}%
where $W:=\Lambda \varphi (Y)-\Phi (\Lambda ).$

We want to bound (\ref{eq:upper-bound}) from above. This expression is
affine in $\varphi $ (recall that $\Phi (u)=\int_{0}^{u}\varphi (t)\,$%
\textrm{d}$t),$ which is a real nondecreasing function with values in $%
[0,1]. $ Since every such function lies in the closed convex hull of the
extreme functions $\varphi =\mathbf{1}_{[p,\infty )}$ for all\footnote{%
Put weight $\theta _{p}=\varphi ^{\prime }(p)\geq 0$ on $\mathbf{1}%
_{[p,\infty )}$ for (almost) every $p>0,$ weight $\theta _{0}=\varphi (0)$
on $\mathbf{1}_{[0,\infty )}\equiv 1\mathbf{,}$ and the remaining weight $%
\theta _{\infty }=1-\int_{0}^{\infty }\varphi ^{\prime }(p)\mathrm{d}%
p-\varphi (0)=1-\varphi (\infty )\geq 0$ on $\mathbf{1}_{[\infty ,\infty
)}\equiv 0;$ cf. Manelli and Vincent (2007, Lemma 4), where it is also shown
how the one-good result (\ref{eq:one good}) easily follows from this claim.} 
$0\leq p\leq \infty ,$ it suffices to bound (\ref{eq:upper-bound}) for these
extreme functions.

Consider such an extreme $\varphi =\mathbf{1}_{[p,\infty )}$ with $p\geq 0;$
then $\Phi (u)=\int_{0}^{u}\varphi (t)\,\mathrm{d}t=\max \{u-p,0\}.$
Substituting in the definition of $W$ yields 
\begin{equation*}
W=\left\{ 
\begin{array}{lcl}
\Lambda -(\Lambda -p)=p, &  & \text{if }Y\geq p\text{ and }Z\geq p, \\ 
Z-0=Z, &  & \text{if }Y\geq p\text{ and }Z<p, \\ 
0-0=0, &  & \text{if }Y<p.%
\end{array}%
\right.
\end{equation*}%
Thus%
\begin{eqnarray*}
\mathbb{E}\left[ W\right] &=&p\,\mathbb{P}\left[ Y\geq p\right] \,\mathbb{P}%
\left[ Z\geq p\right] +\mathbb{P}\left[ Y\geq p\right] \,\mathbb{E}\left[ Z\,%
\mathbf{1}_{Z<p}\right] \\
&=&\mathbb{P}\left[ Y\geq p\right] (\,\mathbb{E}\left[ p\,\mathbf{1}_{Z\geq
p}\right] +\mathbb{E}\left[ Z\,\mathbf{1}_{Z<p}\right] ) \\
&=&G(p)\,\mathbb{E}\left[ \min \{Z,p\}\right]
\end{eqnarray*}%
(we have used the fact that $Y$ and $Z$ are independent and $\min \{Z,p\}=p\,%
\mathbf{1}_{Z\geq p}+Z\,\mathbf{1}_{Z<p}$). Together with $\mathbb{E}\left[
\varphi (Y)\right] =\mathbb{E}\left[ \mathbf{1}_{Y\in \lbrack p,\infty )}%
\right] =\mathbb{P}\left[ Y\geq p\right] =G(p)$, (\ref{eq:upper-bound})
becomes%
\begin{equation}
\mathbb{E}\left[ s(Y,Z)\right] \leq 2r-2rG(p)+2G(p)\,\mathbb{E}\left[ \min
\{Z,p\}\right] )=2(r+\zeta (p))  \label{eq:upper-bound1}
\end{equation}%
where we put $\zeta (p):=G(p)\left( \mathbb{E}\left[ \min \{Z,p\}\right]
-r\right) $. If $p\geq r$ then%
\begin{eqnarray*}
\mathbb{E}\left[ \min \{Z,p\}\right] &=&\int_{0}^{\infty }\mathbb{P}\left[
\min \{Z,p\}\geq u\right] \,\mathrm{d}u=\int_{0}^{p}\mathbb{P}\left[ Z\geq u%
\right] \,\mathrm{d}u \\
&=&\int_{0}^{p}G(u)\,\mathrm{d}u\leq \int_{0}^{r}1\,\mathrm{d}u+\int_{r}^{p}%
\frac{r}{u}\,\mathrm{d}u=r+r\ln \left( \frac{p}{r}\right) ,
\end{eqnarray*}%
where the inequality follows from $G(u)\leq 1$ and $G(u)\leq r/u$ (because $%
r=\sup_{u\geq 0}u\cdot G(u)).$ Therefore%
\begin{equation*}
\zeta (p)\leq G(u)r\ln \left( \frac{p}{r}\right) \leq \frac{r}{p}r\ln \left( 
\frac{p}{r}\right) =r\frac{\ln q}{q},
\end{equation*}%
where $q:=p/r\geq 1.$ Since $\max_{q}(\ln q)/q=1/e$ (attained at $q=e),$ it
follows that $\zeta (p)\leq r/e$ for all $p\geq r.$ If $p\leq r$ then $\zeta
(p)\leq 0$ (since $\mathbb{E}\left[ \min \{Z,p\}\right] \leq p\leq r$), and
so altogether $\zeta (p)\leq r/e$ for all $p\geq 0.$ Therefore $\mathbb{E}%
\left[ s(Y,Z)\right] \leq 2r(1+1/e)$ (recall (\ref{eq:upper-bound1})), which
completes the proof.
\end{proof}

\bigskip

The auxiliary result that we have used is:

\begin{lemma}
\label{l:constrained-R1}Let $X$ be a one-good random valuation that takes
values $X\geq x_{0}$ for some $x_{0}\geq 0.$ Then for every IC and IR
mechanism $\mu =(q,s)$ we have 
\begin{equation}
R(\mu ;X)=\mathbb{E}\left[ s(X)\right] \leq (1-q(x_{0}))~\text{\textsc{Rev}}%
(X)+s(x_{0}).  \label{eq:q0-b0}
\end{equation}
\end{lemma}

\begin{proof}
The function $q$ is nondecreasing (because $q$ is the derivative of the
buyer payoff function $b,$ which is convex), and so $q(x)\geq q(x_{0})$ for
all $x\geq x_{0}.$

If $q(x_{0})=1$ then $q(x)=1$ for all $x\geq x_{0},$ hence $s(x)=s(x_{0})$
for all $x\geq x_{0}$ by IC; therefore $\mathbb{E}\left[ s(X)\right]
=s(x_{0})$ and (\ref{eq:q0-b0}) holds as an equality.

If $q(x_{0})<1$ then we define a new mechanism by rescaling $q$ so that it
uses the full range from $0$ to $1$ (instead of $q(x_{0})$ to $1).$
Specifically, define $\hat{\mu}=(\hat{q},\hat{s})$ by $\hat{q}%
(x):=(q(x)-q(x_{0}))/\lambda $ and $\hat{s}(x):=(s(x)-s(x_{0}))/\lambda $,
where $\lambda :=1-q(x_{0})>0.$ It is immediate to verify that $\hat{\mu}$
is an IC and IR mechanism (for IC, $[\hat{q}(x)\cdot x-\hat{s}(x)]-[\hat{q}(%
\tilde{x})\cdot x-\hat{s}(\tilde{x})]=\left( [q(x)\cdot x-s(x)]-[q(\tilde{x}%
)\cdot x-s(\tilde{x})]\right) /\lambda \geq 0;$ for IR, the resulting buyer
payoff function $\hat{b}$ satisfies $\hat{b}(x_{0})=\hat{q}(x_{0})\cdot
x_{0}-\hat{s}(x_{0})=0).$ Therefore \textsc{Rev}$(X)\geq \mathbb{E}\left[ 
\hat{s}(X)\right] =(\mathbb{E}\left[ s(X)\right] -s(x_{0}))/\lambda ;$
multiplying by $\lambda $ yields (\ref{eq:q0-b0}).
\end{proof}

\subsection{Some Comments on Decomposition\label{s:a-comments}}

We provide here a number of remarks related to the decompositions of
Theorems \ref{th:1/2}, \ref{th:73}, and \ref{th:decomposition}.

\noindent \textbf{Remarks.} \emph{(a) }In the proof of Theorem \ref{th:1/2}
in Section \ref{s:independent}: For every fixed $z,$ applying the
one-dimensional mechanism $\mu ^{z}$ to the whole range of $Y,$ rather than
to $Y\geq z,$ yields $\mathbb{E[}s(Y,z)]\leq \allowbreak $\textsc{Rev}$(Y)+z$
(recall that $s(y,z)\leq s^{z}(y)+z),$ and so, taking expectation over the
values $z$ of $Z,$ and then maximizing over the mechanisms $\mu ,$ we get%
\footnote{%
When $Y$ and $Z$ are not necessarily independent, this becomes \textsc{Rev}$%
(Y,Z)\leq \mathbb{E}[$\textsc{Rev}$(Y|Z)]+\mathbb{E}\left[ Z\right] ,$ where 
$(Y|Z)$ is the random variable $Y$ conditional on the value of $Z,$ and the
expectation is over (the values of) $Z.$} \textsc{Rev}$(Y,Z)\leq $\textsc{Rev%
}$(Y)+\mathbb{E}\left[ Z\right] $\textsc{.} Unfortunately, this inequality
does not suffice: $\mathbb{E}\left[ Z\right] $ may well be infinite, even
when \textsc{Rev}$(Z)$ is finite (as is the case, e.g., for the
Equal-Revenue ($\mathtt{ER}$) distribution (defined in Section \ref%
{s:sep-bun}). This explains the need to split the domain into the two
regions, $Y\geq Z$ and $Y\leq Z,$ which allows us to bound the resulting
expectation terms (see (\ref{eq:e(min)-le-rev(max)})).

\emph{(b) }The proof of Theorem \ref{th:1/2} also implies that \textsc{Rev}$%
(Y)+$\textsc{Rev}$(Z)\geq \mathbb{E}\left[ \min \{Y,Z\}\right] $ (take
expectation of (\ref{eq:e(min)-le-rev(max)}) over the values $z$ of $Z,$
interchange $Y$ and $Z$, and add the two resulting inequalities). Thus,
while in the single-good case one \emph{cannot} guarantee any positive
fraction of the expected value as revenue (take again the $\mathtt{ER}$
distribution, with infinite expectation and revenue $1),$ in the case of two
independent goods one can at least guarantee the expectation of the minimum
of the values of the two goods. A mechanism that yields a revenue of $%
\mathbb{E}\left[ \min \{Y,Z\}\right] $ consists of posting the random prices 
$p_{1}$ for the $y$ good and $p_{2}$ for the $z$ good, where $p_{1}$ and $%
p_{2}$ are independent random variables, $p_{1}$ is distributed like $Z,$
and $p_{2}$ is distributed like $Y$; this is a randomized separate mechanism.%
\footnote{%
The inequality \textsc{Rev}$(Y)+\text{\textsc{Rev}}(Z)\geq \mathbb{E}\left[
\min \{Y,Z\}\right] $ is tight, as it becomes an equality when $Y,Z$ are
i.i.d.-$\mathrm{ER}$ goods. It does not hold when $Y$ and $Z$ are not
independent (for an extreme case take the fully correlated case with $Y=Z$
being $\mathrm{ER});$ the correct inequality here is $\mathbb{E}\left[ \text{%
\textsc{Rev}}(Y|Z)\right] +\mathbb{E}\left[ \text{\textsc{Rev}}(Z|Y)\right]
\geq \mathbb{E}\left[ \min \{Y,Z\}\right] $. All this generalizes to any $%
k\geq 2$ independent goods, where we obtain $\sum_{i}$\textsc{Rev}$%
(X_{i})\geq \mathbb{E}\left[ (m-1)X^{(m)}\right] $ for every $m=1,2,...,k$
(of course, only $m\geq 2$ matters) with $X^{(m)}$ denoting the $m$-th order
statistic of $X_{1},...,X_{k}$ (thus $X^{(1)}=\max_{i}X_{i}$ and $%
X^{(k)}=\min_{i}X_{i}).$}

\emph{(c)} The decomposition of Section \ref{s:decomposition} holds in more
general setups than the totally additive valuation of this paper (where the
value to the buyer of the outcome $q\in \lbrack 0,1]^{k}$ is $%
\sum_{i}q_{i}x_{i}$). Indeed, consider an abstract mechanism-design problem
with a set of alternatives $A$, valued by the buyer according to a function $%
w:A\rightarrow \mathbb{R}_{+}^{k}$ (that he knows, whereas the seller knows
only that the function $w$ is drawn from a certain distribution); assume
also that results such as those in Proposition \ref{p:sub-dom} hold. If the
set of alternatives $A$ is in fact a product $A=A_{1}\times A_{2}$ with the
valuation additive between the two sets, i.e., $%
w(a_{1},a_{2})=w_{1}(a_{1})+w_{2}(a_{2})$, with $w_{1}$ distributed
according to $Y$ and $w_{2}$ according to $Z$, then Theorem \ref%
{th:decomposition} holds as stated. The proof now uses \textsc{$\mathrm{Val}$%
}$(Z)=\mathbb{E}[\sup_{a_{2}\in A_{2}}w_{2}(a_{2})]$ (which, in our case,
where $A_{2}=[0,1]^{k_{2}}$ and $w_{2}(q)=\sum_{j}q_{j}z_{j}$, is indeed 
\textsc{$\mathrm{Val}$}$(Z)=\mathbb{E}(\sum_{j}Z_{j})$ since $%
\sup_{q}w_{2}(q)=\sum_{j}z_{j}$).

\subsection{Equal Revenue ($\mathtt{ER}$) Goods\label{s:ER}}

In this appendix we prove the claims of Proposition \ref{p:ER} concerning $%
\mathtt{ER}$ goods: Lemma \ref{l:rER}, Propositions \ref{l:brev-ER2} and \ref%
{l:brev-ERk}, and Corollary \ref{c:equalizek}.

\begin{lemma}
\label{l:rER}Let $X$ be a one-good random valuation. Then \textsc{Rev}$%
(X)\leq r$ if and only if $X$ is stochastically dominated by $rV$ where $V$
is an $\mathtt{ER}$ valuation.
\end{lemma}

\begin{proof}
By (\ref{eq:one good}), \textsc{Rev}$(X)\leq r$ if and only if $\mathbb{P}%
[X\geq \nolinebreak p]\leq r/p$ for every $p\geq 0$; this inequality matters
only for $p>r,$ for which $r/p=\mathbb{P}\left[ rV\geq p\right] .$
\end{proof}

\bigskip

Next we compute the distribution of a weighted sum of two independent $%
\mathtt{ER}$ distributions.

\begin{lemma}
\label{aER+bER} Let $V_{1},V_{2}$ be i.i.d.-$\mathtt{ER}$ and let $\alpha
,\beta >0$. Then%
\begin{equation*}
\mathbb{P}\left[ \alpha V_{1}+\beta V_{2}\geq z\right] =\frac{\alpha \beta }{%
z^{2}}\ln \left( 1+\frac{z^{2}-(\alpha +\beta )z}{\alpha \beta }\right) +%
\frac{\alpha +\beta }{z}
\end{equation*}%
for $z\geq \alpha +\beta $, and $\mathbb{P}\left[ \alpha V_{1}+\beta
V_{2}\geq z\right] =1$ for $z\leq \alpha +\beta $.
\end{lemma}

\begin{proof}
Let $Z=\alpha V_{1}+\beta V_{2}.$ For $z\leq \alpha +\beta $ we have $%
\mathbb{P}\left[ Z\geq z\right] =1$ since $V_{i}\geq 1.$ For $z>\alpha
+\beta $ we get%
\begin{eqnarray*}
\mathbb{P}\left[ Z\geq z\right] &=&\int f(x)\left( 1-F\left( \frac{z-\alpha x%
}{\beta }\right) \right) \mathrm{d}x \\
&=&\int_{1}^{(z-\beta )/\alpha }\frac{1}{x^{2}}\frac{\beta }{z-\alpha x}\,%
\mathrm{d}x+\int_{(z-\beta )/\alpha }^{\infty }\frac{1}{x^{2}}1\,\mathrm{d}x
\\
&=&\frac{\beta }{z}\left[ \frac{\alpha }{z}\ln x-\frac{\alpha }{z}\ln \left( 
\frac{z}{\alpha }-x\right) -\frac{1}{x}\right] _{1}^{(z-\beta )/\alpha }+%
\frac{\alpha }{z-\beta } \\
&=&\frac{\alpha \beta }{z^{2}}\left( \ln \left( \frac{z}{\beta }-1\right)
+\ln \left( \frac{z}{\alpha }-1\right) \right) -\frac{\alpha \beta }{%
z(z-\beta )}+\frac{\beta }{z}+\frac{\alpha }{z-\beta } \\
&=&\frac{\alpha \beta }{z^{2}}\ln \left( 1+\frac{z^{2}-(\alpha +\beta )z}{%
\alpha \beta }\right) +\frac{\alpha +\beta }{z}\,,
\end{eqnarray*}%
completing the proof.
\end{proof}

\bigskip

Weighted sums of independent $\mathtt{ER}$ distributions are used in the
proof of Proposition \ref{p:sep>bun} (see Section \ref{s:sep-bun} and recall
Lemma \ref{l:rER}). What we will show now (Lemma \ref{l:ER-equalize2} and
Corollary \ref{c:equalizek}) is that moving the weights in the direction of
equalizing them yields stochastic domination.

\begin{lemma}
\label{l:ER-equalize2}Let $V_{1},V_{2}$ be i.i.d.-$\mathtt{ER}$ and let $%
\alpha ,\beta ,a^{\prime },\beta ^{\prime }>0.$ If $\alpha +\beta =\alpha
^{\prime }+\beta ^{\prime }$ and\footnote{%
Equivalently, $\alpha ^{\prime },\beta ^{\prime }$ are closer to one another
than $\alpha ,\beta $ are; i.e., $|\alpha ^{\prime }-\beta ^{\prime }|\leq
|\alpha -\beta |.$} $\alpha \beta \leq \alpha ^{\prime }\beta ^{\prime }$
then $\alpha V_{1}+\beta V_{2}$ is stochastically dominated by $\alpha
^{\prime }V_{1}+\beta ^{\prime }V_{2}.$
\end{lemma}

\begin{proof}
Let $Z=\alpha V_{1}+\beta V_{2}$ and $Z^{\prime }=\alpha ^{\prime
}V_{1}+\beta ^{\prime }V_{2},$ and put $\gamma =\alpha +\beta =\alpha
^{\prime }+\beta ^{\prime }.$ Using Lemma \ref{aER+bER}, for $z\leq \gamma $
we have $\mathbb{P}[Z\geq z]=\mathbb{P}[Z^{\prime }\geq z]=1,$ and for $%
z>\gamma $ we get%
\begin{eqnarray*}
\mathbb{P}[Z\geq z] &=&\frac{\alpha \beta }{z^{2}}\ln \left( 1+\frac{%
z^{2}-\gamma z}{\alpha \beta }\right) +\frac{\gamma }{z} \\
&\leq &\frac{\alpha ^{\prime }\beta ^{\prime }}{z^{2}}\ln \left( 1+\frac{%
z^{2}-\gamma z}{\alpha ^{\prime }\beta ^{\prime }}\right) +\frac{\gamma }{z}=%
\mathbb{P}[Z^{\prime }\geq z],
\end{eqnarray*}%
since $t\ln (1+1/t)$ is increasing in $t>0$, and $\alpha \beta
/(z^{2}-\gamma z)\leq $\allowbreak $\alpha ^{\prime }\beta ^{\prime
}/(z^{2}-\gamma z)$ by our assumption that $\alpha \beta \leq \alpha
^{\prime }\beta ^{\prime }$ together with $z>\gamma $.
\end{proof}

\bigskip

\begin{corollary}
\label{c:equalizek} Let $V_{1},V_{2},...,V_{k}$ be i.i.d.-$\mathtt{ER,}$ let 
$r_{1},r_{2},...,r_{k}\geq 0,$ and put $\bar{r}=(1/k)\sum_{i=1}^{k}r_{i}$
for the average of the $r_{i}.$ Then $\sum_{i=1}^{k}r_{i}V_{i}$ is
stochastically dominated by $\sum_{i=1}^{k}\bar{r}V_{i}.$
\end{corollary}

\begin{proof}
If, say, $r_{1}<\bar{r}<r_{2},$ then Lemma \ref{l:ER-equalize2} above
implies that $r_{1}V_{1}+r_{2}V_{2}$ is stochastically dominated by $\bar{r}%
V_{1}+r_{2}^{\prime }V_{2}$, where $r_{2}^{\prime }=r_{1}+r_{2}-\bar{r}>0,$
and so\footnote{%
Recall footnote \ref{ft:stoch-dom-convolution}: stochastic dominance is
closed under convolutions.} $\sum_{i=1}^{k}r_{i}V_{i}$ is stochastically
dominated by $\bar{r}V_{1}+r_{2}^{\prime }V_{2}+\sum_{i=3}^{k}r_{i}V_{i}.$
Continue this way until all coefficients become $\bar{r}.$
\end{proof}

\bigskip

We now calculate the revenue obtainable from bundling two independent $%
\mathtt{ER}$ goods. Recall that $w\approx 0.278$ is the solution of the
equation $we^{w+1}=1,$ or $\ln w+w=-1.$

\begin{proposition}
\label{l:brev-ER2}Let $V_{1},V_{2}$ be i.i.d.-$\mathtt{ER}$. Then 
\begin{equation*}
\text{\textsc{BRev}}(V_{1},V_{2})=\text{\textsc{Rev}}(V_{1}+V_{2})=2(w+1)%
\approx 2.56.
\end{equation*}
\end{proposition}

\begin{proof}
Using Lemma \ref{aER+bER} with $\alpha =\beta =1$ yields $p~\mathbb{P}\left[
V_{1}+V_{2}\geq p\right] =p^{-1}\ln (1+p^{2}-2p)+2=2p^{-1}\ln (p-1)+2,$
which attains its maximum of $2w+2$ at $p=1+1/w$ (i.e., $1/(p-1)=w).$
\end{proof}

\bigskip

We estimate the bundling revenue from $k$ independent $\mathtt{ER}$ goods.

\begin{proposition}
\label{l:brev-ERk}There exist constants $c_{1}>0$ and $c_{2}<\infty $ such
that for any $k\geq 2$ and $k$ i.i.d.-$\mathtt{ER}$ goods $%
V_{1},V_{2},...,V_{k}$, 
\begin{equation*}
c_{1}k\log k\leq \text{\textsc{BRev}}(V_{1},V_{2},...,V_{k})=\text{\textsc{%
Rev}}(V_{1}+...+V_{k})\leq c_{2}k\log k.
\end{equation*}
\end{proposition}

\begin{proof}
For a one-dimensional random variable $X$ and a constant $M,$ write $%
X^{M}:=\min \{X,M\}$ for $X$ truncated at $M.$ When $V$ is $\mathtt{ER}$ and 
$M\geq 1$ it is immediate to compute $\mathbb{E}\left[ V^{M}\right] =\ln M+1$
and $\mathrm{Var}(V^{M})\leq 2M$.

$\bullet $ \emph{Lower bound}: For every $p,M>0$ we have \textsc{Rev}$%
(\sum_{i}V_{i})\geq p\cdot \mathbb{P}\left[ \sum_{i}V_{i}\geq p\right] \geq
p\cdot \mathbb{P}\left[ \sum_{i}V_{i}^{M}\geq p\right] $.

When $M=k\ln k$ and $p=(k\ln k)/2$ we get $(k\mathbb{E}\left[ V^{M}\right]
-p)/\sqrt{k\mathrm{Var}(V^{M})}\geq \sqrt{\ln k/8}$, and so $p$ is at least $%
\sqrt{\ln k/8}$ standard deviations below the mean of $%
\sum_{i=1}^{k}V_{i}^{M}$. Therefore, by Chebyshev's inequality, $\mathbb{P}%
\left[ \sum_{i=1}^{k}V_{i}^{M}\geq p\right] \geq 1-8/\ln k\geq 1/2$ for all $%
k$ large enough, and then \textsc{Rev}$(\sum_{i=1}^{k}V_{i})\geq p\cdot
1/2=k\ln k/4$.

$\bullet $ \emph{Upper bound}: We need to bound $\sup_{p\geq 0}p\cdot 
\mathbb{P}\left[ \sum_{i=1}^{k}V_{i}\geq p\right] $.

Consider two cases for $p.$ If $p\leq 6k\ln k$ then $p\cdot \mathbb{P}\left[
\sum_{i=1}^{k}V_{i}\geq p\right] \leq p\leq 6k\ln k$.

If $p\geq 6k\ln k,$ then, taking $M=p$, we have 
\begin{equation}
p\cdot \mathbb{P}\left[ \sum_{i=1}^{k}V_{i}\geq p\right] \leq p\cdot \mathbb{%
P}\left[ \sum_{i=1}^{k}V_{i}^{p}\geq p\right] +p\cdot \mathbb{P}\left[
V_{i}>p\text{ for some }1\leq i\leq k\right] .  \label{eq:z_large}
\end{equation}%
The second term on the right-hand side is at most $p\cdot k\cdot
(1-F_{V}(p))=k$ (since $F_{V}(p)=1-1/p$). To estimate the first term, we
again use Chebyshev's inequality: $\sum_{i=1}^{k}V_{i}^{p}$ has mean $k(\ln
p+1)$ and standard deviation $\sqrt{2kp}$. When $k$ is large enough we have $%
p/(k(\ln p+1))\geq 2$ (recall that $p\geq 6k\ln k$), hence $(p-k(\ln p+1))/%
\sqrt{2kp}\geq (p/2)/\sqrt{2kp}=\sqrt{p/(8k)},$ and so $p$ is at least $%
\sqrt{p/(8k)}$ standard deviations above the mean of $%
\sum_{i=1}^{k}V_{i}^{p} $. Therefore $p\cdot \mathbb{P}\left[
\sum_{i=1}^{k}V_{i}^{p}\geq p\right] \leq p\cdot (8k)/p=8k$, which implies $%
p\cdot \mathbb{P}\left[ \sum_{i=1}^{k}V_{i}\geq p\right] \leq 9k$ by (\ref%
{eq:z_large}).

Altogether, \textsc{Rev}$(\sum_{i=1}^{k}V_{i})\leq \max \{6k\ln k,9k\}=6k\ln
k$ for all $k$ large enough.
\end{proof}

\bigskip

\noindent \textbf{Remark.} A more precise analysis, based on a Generalized
Central Limit Theorem\ (see, e.g., Zaliapin, Kagan, and Schoenberg 2005),
shows that \textsc{Rev}$(\sum_{i=1}^{k}V_{i})/(k\ln k)$ converges to $1$ as $%
k\rightarrow \infty .$ Indeed, the sequence $%
(\sum_{i=1}^{k}V_{i}-b_{k})/a_{k}$ with $a_{k}=k\pi /2$ and\footnote{\label%
{ftn:O}We use the standard computer science notations: $f(k)=\mathrm{O}%
(g(k)) $ means that there exists a constant $c<\infty $ such that $f(k)\leq
cg(k)$ for all $k,$ and $f(k)=\Omega (g(k))$ means that there exists a
constant $c>0 $ such that $f(k)\geq cg(k)$ for all $k.$ Also, $f(k)=\Theta
(g(k))$ means that $f(k)=\mathrm{O}(g(k)$ and $f(k)=\Omega (g(k))$ both
hold; i.e., there exist $c_{1}>0$ and $c_{2}<\infty $ such that $%
c_{1}g(x)\leq f(k)\leq c_{2}g(k)$ for all $k.$} $b_{k}=k\ln k+\Theta (k)$
converges in distribution to the Cauchy distribution as $k\rightarrow \infty 
$. The revenue from a Cauchy distribution can easily be shown to be bounded
(by $1/\pi ;$ use (\ref{eq:one good})), and so it follows that \textsc{Rev}$%
(\sum_{i=1}^{k}V_{i})=k\ln k+\Theta (k).$

\bigskip

As a corollary, we get that separate selling may yield no more than a
fraction of the order of $1/\log k$ of the optimal revenue.

\begin{corollary}
\label{c:sep-k-er}There exists a constant $c<\infty $ such that for any $%
k\geq 2$ and $k$ i.i.d.-$\mathtt{ER}$ goods $V_{1},V_{2},...,V_{k}$%
\begin{equation*}
\text{\textsc{SRev}}(V_{1},V_{2},...,V_{k})\leq \frac{c}{\log k}\text{%
\textsc{Rev}}(V_{1},V_{2},...,V_{k}).
\end{equation*}
\end{corollary}

\begin{proof}
We have \textsc{Rev}$(V_{1},...,V_{k})\geq $\textsc{BRev}$%
(V_{1},...,V_{k})\geq c_{1}\log k\cdot k=c_{1}\log k\cdot \text{\textsc{SRev}%
}(V_{1},...,V_{k})$ by Proposition \ref{l:brev-ERk}, and \textsc{SRev}$%
(V_{1},...,V_{k})=k$ (because \textsc{Rev}$(V_{i})=1$).
\end{proof}

\subsection{Separate vs. Bundled Selling\label{s:ap-sep-bun}}

We start with an example showing that the $1/k$ bound for $k$ independent
goods of Proposition \ref{p:bun>sep} (i) is tight.

\begin{example}
\label{ex:1/k} \textsc{BRev}$(X_{1},...,X_{k})=(1/k+\varepsilon )\cdot $%
\textsc{SRev}$(X_{1},...,X_{k})$: Take a large $M$ and let $X_{i}$ have
support $\{0,M^{i}\}$ with $\mathbb{P}[X_{i}=M^{i}]=M^{-i}$. Then \textsc{Rev%
}$(X_{i})=1$ and so \textsc{SRev}$(X_{1},...,X_{k})=k$, while \textsc{BRev}$%
(X_{1},...,X_{k})$ is easily seen to be at most $\max_{i}M^{i}\cdot
(M^{-i}+\cdots +M^{-k})\leq 1+1/(M-1)$. Because \textsc{SRev}$\leq $\textsc{%
Rev }this also shows that bundling may yield no more than a $1/k$ fraction
of the optimal revenue: \textsc{BRev}$(X_{1},...,X_{k})\leq (1/k+\varepsilon
)\cdot $\textsc{Rev}$(X_{1},...,X_{k})$\textsc{.}
\end{example}

\bigskip

We next prove that a \textsc{GFOR} of the order of $1/k$ is tight.

\begin{lemma}
\label{B-1/k} There exists a constant $c>0$ such that for any $k\geq 2$ and
any $k$ independent goods $X_{1},X_{2},...,X_{k},$ 
\begin{equation*}
\text{\textsc{BRev}}(X_{1},X_{2},...,X_{k})\geq \frac{c}{k}~\text{\textsc{Rev%
}}(X_{1},X_{2},...,X_{k}).
\end{equation*}
\end{lemma}

\begin{proof}
For $k$ a power of two, we use (cf. the proof of Theorem \ref{th:srev-k} in
Section \ref{s:proofs-thm3-4}) the decomposition of (\ref{eq:k/2}) to obtain
by induction, starting from \textsc{Rev}$(X_{1})=$\textsc{BRev}$(X_{1}),$
the inequality \textsc{Rev}$(X_{1},...,X_{k})\leq (3k-2)\cdot $\textsc{BRev}$%
(X_{1},...,X_{k})$ (the induction step uses the fact that the bundled
revenue from a subset of the goods is at most the bundled revenue from all
of them, since all the $X_{i}$ are nonnegative). Again, when $k$ is not a
power of $2$ we can pad to the next power of $2$ with goods that have value
identically zero, which at most doubles $k$.
\end{proof}

\bigskip

Next, we consider i.i.d. goods, where better bounds can be obtained: the
bundling revenue cannot be much smaller than the separate revenue.

\begin{lemma}
\label{l:2/3}For any two i.i.d. goods $X_{1},X_{2},$ 
\begin{equation*}
\text{\textsc{BRev}}(X_{1},X_{2})\geq \frac{2}{3}~\text{\textsc{SRev}}%
(X_{1},X_{2}).
\end{equation*}
\end{lemma}

\begin{proof}
Let $F$ be the distribution of the $X_{i},$ let $p$ be the optimal one-good
price for $F,$ and put $\alpha :=1-F(p)$; thus, \textsc{Rev}$(X_{i})=p\alpha 
$. If $\alpha \leq 2/3$ then the bundling mechanism can offer a price of $p,$
and then the probability that the bundle will be sold is at least the
probability that one of the goods by itself has value $p$, which is $2\alpha
-\alpha ^{2}=\alpha (2-\alpha )\geq 4\alpha /3$; the revenue is then at
least $p\cdot 4\alpha /3=(4/3)$\textsc{Rev}$(X_{i})$. If $\alpha \geq 2/3$
then the bundling mechanism can offer a price of $2p$, and then the
probability that it will be accepted is at least the probability that both
goods will get a value of at least $p$, which is $\alpha ^{2}$; the revenue
is then $2p\cdot \alpha ^{2}\geq (4/3)p\alpha =(4/3)$\textsc{Rev}$(X_{i}).$
In both cases the bundling revenue was at least $(4/3)\text{\textsc{Rev}}%
(X_{i})=(2/3)$\textsc{SRev}$(X_{1},X_{2})$.
\end{proof}

\bigskip

This $2/3$ bound is tight.

\begin{example}
\label{ex:2/3}\textsc{BRev}$(X_{1},X_{2})=(2/3)\cdot $\textsc{SRev}$%
(X_{1},X_{2})$: Let $X_{i}$ have support $\{0,1\}$ with $\mathbb{P}%
[X_{i}=1]=2/3$; then \textsc{Rev}$(X_{i})=2/3$ while \textsc{BRev}$%
(X_{1},X_{2})=8/9$ (which is obtained both at price $1$ and at price $2$).%
\footnote{%
It can be checked that the optimal revenue is attained here by selling
separately, i.e., \textsc{Rev}$(X_{1},X_{2})=$\textsc{SRev}$%
(X_{1},X_{2})=4/3 $.}
\end{example}

\subsection{Many I.I.D. Goods\label{s:k-iid}}

It is well known that when the goods are independent and identically
distributed, and their number $k$ tends to infinity, then the bundling
revenue approaches the optimal revenue. Even more, essentially all the
buyer's surplus can be extracted by selling optimally the bundle of all
goods. The logic is quite simple: the law of large numbers tells us that
there is almost no uncertainty about the sum of many i.i.d. random
variables, and so the seller essentially knows this sum and may ask for it
as the bundle price. For completeness we state this result and provide a
short proof, which also covers the case where the expectation is infinite.

\begin{theorem}[Armstrong 1999, Bakos and Brynjolfsson 1999]
Let $X_{i}$ be i.i.d. one-good random valuations. Then%
\begin{equation*}
\lim_{k\rightarrow \infty }\frac{\text{\textsc{BRev}}(X_{1},X_{2},...,X_{k})%
}{k}=\lim_{k\rightarrow \infty }\frac{\text{\textsc{Rev}}%
(X_{1},X_{2},...,X_{k})}{k}=\mathbb{E}\left[ X_{1}\right] .
\end{equation*}
\end{theorem}

\begin{proof}
We always have \textsc{BRev}$(X_{1},...,X_{k})\leq $\textsc{Rev}$%
(X_{1},...,X_{k})\leq k\mathbb{E}\left[ X_{1}\right] $ (the second
inequality follows from $s(x)=q(x)\cdot x-b(x)\leq \sum_{i}x_{i}$ by IR).
Let us assume first that the $X_{i}$ have finite expectation and finite
variance. In this case if we charge a price of $(1-\varepsilon )k\mathbb{E}%
\left[ X_{1}\right] $ for the bundle, then, by Chebyshev's inequality, the
probability that the bundle will not be bought is at most \textsc{Var}$%
(X_{1})/(\varepsilon ^{2}\mathbb{E}\left[ X_{1}\right] \sqrt{k})$, and this
goes to zero as $k$ increases.

If the expectation or variance is infinite, then consider the truncated
distribution where values above a certain $M$ are replaced by $M$, which has
finite expectation and variance. We can choose the finite $M$ so as to bring
the expectation of the truncated distribution as close as we desire to the
original one (including as high as we desire, if the original distribution
has infinite expectation).
\end{proof}

\bigskip

Despite the apparent strength of this result, it does not provide any
guarantees for any \emph{fixed value} of $k$. Indeed, we now show that for
every large enough $k$ we have \textsc{GFOR}$($\textsc{bundled;~}$k$ i.i.d.
goods$)\leq 57\%$ (recall that it is at least $1/4$ by Proposition \ref%
{p:bun>sep} (ii)). Recently, Kupfer (2016) obtained the precise value of the
limit of this \textsc{GFOR} as $k$ increases; it turns out to be
approximately $55.9\%.$

\begin{example}
\label{ex:57} \emph{For every }$k$\emph{\ large enough, a one-dimensional
distribution }$F$\emph{\ (which depends on }$k)$ \emph{such that }\textsc{%
BRev}$(X_{1},...,X_{k})\leq 0.57\cdot $\textsc{SRev}$(X_{1},...,X_{k}),$%
\emph{\ and thus }\textsc{BRev}$(X_{1},...,X_{k})\leq 0.57\cdot $\textsc{Rev}%
$(X_{1},...,X_{k})$\emph{, where the }$X_{i}$\emph{\ are i.i.d.-}$F$\emph{\
goods:} For each $k$ consider the distribution $F$ on $\{0,1\}$ with $%
\mathbb{P}[X=1]=c/k$ where $c\approx 1.256$ is the positive solution of $%
1-e^{-c}=2(1-(1+c)e^{-c})$; the revenue from selling a single good is thus $%
c/k$, and so \textsc{SRev}$(X_{1},...,X_{k})=c.$ The bundling mechanism
should clearly offer an integral price. If it offers price $1$ then the
probability of selling is $1-(1-c/k)^{k}$, which converges to $%
1-e^{-c}\approx 0.715$ as $k$ increases. If it offers price $2$ then the
probability of selling is $1-(1-c/k)^{k}-k(c/k)(1-c/k)^{k-1}\rightarrow
1-(1+c)e^{-c},$ and the revenue is twice that, again $\approx 0.715$ in the
limit (recall the equation that $c$ satisfies). If it offers price $3$ then
the probability of selling is $1-(1-c/k)^{k}-k(c/k)(1-c/k)^{k-1}-{\binom{k}{2%
}}(c/k)^{2}(1-c/k)^{k-2}\rightarrow 1-(1+c+c^{2}/2!)e^{-c}\approx 0.13$, and
the revenue is three times that, which is less than $0.715$. For higher
integral prices $m\geq 4$ the probability of selling converges to $%
1-(1+c+...+c^{m-1}/(m-1)!)e^{-c}\leq c^{m}/m!$ (because the corresponding
remainder in the $e^{c}$ series is bounded by $e^{c}c^{m}/m!),$ and the
revenue is thus $\leq c^{m}/(m-1)!$, which is even smaller. Therefore for
all large enough $k$ the optimal bundle price is either $1$ or $2,$ and 
\textsc{BRev}$(X_{1},...,X_{k})/$\textsc{SRev}$(X_{1},...,X_{k})$ is close
to $(1-e^{-c})/c\approx 0.569.$
\end{example}

\subsection{When Bundling Is Optimal\label{s:brev=rev}}

In this appendix we prove Theorem \ref{th:brev-opt}, stated in Section \ref%
{sus:indep-bundled}: for two i.i.d. goods, if the one-good distribution
satisfies condition (\ref{eq:3/2}), then bundling is optimal.

\bigskip

\begin{proof}[Proof of Theorem \protect\ref{th:brev-opt}]
Let $X=(Y,Z)$ where $Y,Z$ are i.i.d. with cumulative distribution $F$ and
probability density function $f$ that satisfies (\ref{eq:3/2}). We will show
that for every IC and IR mechanism $\mu $ there is a bundled mechanism $\hat{%
\mu}$ that yields at least as much revenue, i.e., $R(\hat{\mu};X)\geq R(\mu
;X);$ this proves that \textsc{Rev}$(X)=~$\textsc{BRev}$(X).$

Let $\mu =(q,s)$ be an IC and IR mechanism with buyer payoff function $b$;
as in the proof of Theorem \ref{th:73} in Appendix \ref{s:proof-iid}, we
assume without loss of generality that the mechanism is symmetric and
satisfies NPT. Therefore%
\begin{eqnarray*}
\mathbb{E}\left[ s(Y,Z)\right] &=&\mathbb{E}\left[
Yq_{1}(Y,Z)+Zq_{2}(Y,Z)-b(Y,Z)\right] \\
&=&\mathbb{E}\left[ 2Yq_{1}(Y,Z)-b(Y,Z)\right] ,
\end{eqnarray*}%
because $\mathbb{E}\left[ Zq_{2}(Y,Z)\right] =\mathbb{E}\left[ Zq_{1}(Z,Y)%
\right] =\mathbb{E}\left[ Yq_{1}(Y,Z)\right] $ by symmetry and then by
interchanging the i.i.d. variables $Y$ and $Z.$ Truncating at $M$ therefore
yields (recall that $s$ is nonnegative by NPT)%
\begin{equation*}
R(\mu ;X)=\lim_{M\rightarrow \infty }r_{M}(b),
\end{equation*}%
where 
\begin{equation}
r_{M}(b)%
{\;:=\;}%
\lim_{M\rightarrow \infty }\int_{a}^{M}\int_{a}^{M}\left(
2yb_{y}(y,z)-b(y,z)\right) f(y)f(z)\,\mathrm{d}y\,\mathrm{d}z  \label{eq:rM}
\end{equation}%
(because $q_{1}(y,z)=b_{y}(y,z),$ the derivative of $b(y,z)$ with respect to
its first variable, for almost every $(y,z)$ by Proposition \ref{p:IC
=b-convex}, and the distribution $F$ is continuous).

For each fixed $z$ integrate by parts the $2yb_{y}(y,z)f(y)$ term:%
\begin{eqnarray*}
\int_{a}^{M}2b_{y}(y,z)yf(y)\,\mathrm{d}y &=&\left[ 2b(y,z)yf(y)\right]
_{a}^{M}-\int_{a}^{M}2b(y,z)\left( f(y)+yf^{\prime }(y)\right) \,\mathrm{d}y
\\
&=&2b(M,z)Mf(M)-2b(a,z)af(a) \\
&&-\int_{a}^{M}2b(y,z)\left( f(y)+yf^{\prime }(y)\right) \,\mathrm{d}y.
\end{eqnarray*}%
Substituting this in (\ref{eq:rM}) yields%
\begin{eqnarray}
r_{M}(b) &=&2Mf(M)\int_{a}^{M}b(M,z)f(z)\,\mathrm{d}z  \notag \\
&&+2\int_{a}^{M}\int_{a}^{M}b(y,z)\left( -\frac{3}{2}f(y)-yf^{\prime
}(y)\right) f(z)\,\mathrm{d}y\,\mathrm{d}z  \label{eq:r_M(b)} \\
&&-2af(a)\int_{a}^{M}b(a,z)f(z)\,\mathrm{d}z.  \notag
\end{eqnarray}

Define $\hat{b}(y,z):=b(y+z-a,a)=b(a,y+z-a)$ for every $(y,z)$ with $y,z\geq
a.$ Then $\hat{b}$ is a symmetric convex function on the quadrant $[a,\infty
)^{2},$ it coincides with $b$ on the boundaries $y=a$ and $z=a,$ and is at
least as large as $b$ everywhere\footnote{%
The function $\hat{b}$ is in fact the \emph{smallest} function satisfying
these three properties (i.e., it is a convex function, coincides with $b$ on
the boundary, and is everywhere $\geq b);$ see Hart (2012) for an
interesting observation on this.}: indeed, the convexity of $b$ yields%
\begin{eqnarray}
b(y,z) &\leq &\frac{y-a}{y+z-2a}b(y+z-a,a)+\frac{z-a}{y+z-2a}b(a,y+z-a)
\label{eq:b-le-b-hat} \\
&=&b(y+z-a,a)=\hat{b}(y,z)  \notag
\end{eqnarray}%
for every $(y,z)\in \lbrack a,\infty )^{2}$. Replacing $b$ with $\hat{b}$
can only increase the first and second terms of (\ref{eq:r_M(b)}), since all
the coefficients of $b$ there are nonnegative (use (\ref{eq:3/2})), while it
does not affect the third term. Therefore $r_{M}(b)\leq r_{M}(\hat{b})$ for
all $M$.

Define $\hat{q}(y,z):=(q_{1}(y+z-a,a),q_{1}(y+z-a,a))\in \lbrack 0,1]^{2}$
and $\hat{s}(y,z):=\hat{q}(y,z)\cdot (y,z)-\hat{b}(y,z);$ then $\hat{q}(y,z)$
is a subgradient of\footnote{%
At points of differentiability $\hat{b}_{y}(y,z)=\hat{b}%
_{z}(y,z)=b_{y}(y+z-a,a).$} $\hat{b}$ at $(y,z),$ and so $\hat{\mu}=(\hat{q},%
\hat{s})$ is an IC and IR mechanism (by Proposition \ref{p:IC =b-convex}).
Since $R(\hat{\mu};X)=\lim_{M\rightarrow \infty }r_{M}(\hat{b})$ and $%
r_{M}(b)\leq r_{M}(\hat{b})$ for all $M,$ we get $R(\mu ;X)\leq R(\hat{\mu}%
;X).$ Finally, $\hat{\mu}$ is a bundled mechanism, since $\hat{q}$ and $\hat{%
s}$ are functions of $y+z$ (the corresponding one-good mechanism for $T:=Y+Z$
is $(\tilde{q},\tilde{s})$ with $\tilde{q}(t):=\hat{q}(t-a,a)$ and $\tilde{s}%
(t):=\hat{s}(t-a,a)).$ This completes the proof.
\end{proof}

\subsection{Multiple Buyers\label{s:n-buyers-proofs}}

The generalization from one buyer to $n\geq 1$ buyers is as follows (recall
Section \ref{sus:model}). One seller is selling $k$ goods to $n$ buyers$;$
these goods have no value or cost to the seller. For each buyer $j=1,...,n$
and each good $i=1,...,k,$ buyer $j$'s value for good $i$ is given by a
nonnegative random variable $X_{i}^{j};$ put $X^{j}=(X_{i}^{j})_{i=1,..,k}$
for the $\mathbb{R}_{+}^{k}$-valued random valuation vector of buyer $j,$
and $X_{i}=(X_{i}^{j})_{j=1,...,n}$ for the $\mathbb{R}_{+}^{n}$-valued
vector of values of good $i;$ put also $X=(X_{i}^{j})_{j=1,...,n;i=1,...,k},$
and let $F$ be the joint distribution of all these $kn$ random variables.
The distribution $F$ is commonly known; in addition, each buyer $j$ knows
the realization of his random valuation $X^{j}.$ Finally, the seller as well
as all the buyers are risk-neutral and have quasi-linear utilities, and the
valuation of a set of goods to each buyer is additive.

A (direct)\emph{\ mechanism} $\mu =(q^{j},s^{j})_{j=1,...,n}$ consists of an
allocation function $q^{j}:\mathbb{R}_{+}^{kn}\rightarrow \lbrack 0,1]^{k}$
and a payment function $s^{j}:\mathbb{R}_{+}^{kn}\rightarrow \mathbb{R}$ for
each buyer $j=1,...,n,$ where $\sum_{j=1}^{n}q_{i}^{j}(x)\leq 1$ for every
good $i=1,..,k;$ the payoff of buyer $j$ is $b^{j}(x)=q^{j}(x)\cdot
x^{j}-s^{j}(x),$ and that of the seller is $S(x):=\sum_{j=1}^{n}s^{j}(x).$
Two standard equilibrium notions are used for the $n$-person game among the
buyers (once the mechanism $\mu $ is given): \textquotedblleft \emph{%
dominant strategy}" (\textsc{ds}) and \textquotedblleft \emph{Bayesian Nash}%
\textquotedblright\ (\textsc{bn}), which yield the so-called ex-post and
interim equilibria, respectively. The corresponding conditions are:

\begin{itemize}
\item In the dominant strategy case: \emph{incentive compatibility} (\textbf{%
IC-DS}) requires that%
\begin{equation*}
b^{j}(x)=q^{j}(x)\cdot x^{j}-s^{j}(x)=\max_{\tilde{x}^{j}\in \mathbb{R}%
_{+}^{k}}\left[ q^{j}(\tilde{x}^{j},x^{-j})\cdot x^{j}-s^{j}(\tilde{x}%
^{j},x^{-j})\right]
\end{equation*}%
for every $j=1,...,n$ and $x\in \mathbb{R}_{+}^{kn}$; \emph{individual
rationality} (\textbf{IR-DS}) requires that $b^{j}(x)\geq 0$ for every $j$
and $x\in \mathbb{R}_{+}^{kn}.$

\item In the Bayesian Nash case: \emph{incentive compatibility} (\textbf{%
IC-BN}) requires that\footnote{%
The conditions in the Bayesian Nash case depend on the mechanism $\mu $ and
(the distribution of) the valuations $X,$ whereas in the dominant strategy
case they depend only on the mechanism $\mu .$} 
\begin{equation*}
\bar{b}^{j}(x^{j})%
{\;:=\;}%
\mathbb{E}\left[ b^{j}(X)|X^{j}=x^{j}\right] =\max_{\tilde{x}^{j}\in \mathbb{%
R}_{+}^{k}}\mathbb{E}\left[ q^{j}(\tilde{x}^{j},X^{-j})\cdot x^{j}-s^{j}(%
\tilde{x}^{j},X^{-j})|X^{j}=x^{j}\right]
\end{equation*}%
for every $j=1,...,n$ and $x^{j}\in \mathbb{R}_{+}^{k};$ \emph{individual
rationality} (\textbf{IR-BN}) requires that $\bar{b}^{j}(x^{j})\geq 0$ for
every $j$ and $x^{j}\in \mathbb{R}_{+}^{k}$.
\end{itemize}

Let $R(\mu ;X):=\mathbb{E}\left[ S(X)\right] \equiv \mathbb{E}\left[
\sum_{j=1}^{n}s^{j}(X)\right] $ denote the seller's expected revenue from
the mechanism $\mu $ for the random valuations $X.$ Let \textsc{Rev}$^{DS}(X%
\mathcal{)}$ stand for the maximal revenue obtained from $k$ goods and $n$
buyers with random valuations $X$ using dominant strategy implementation,
i.e., mechanisms that satisfy IC-DS and IR-DS; let \textsc{Rev}$^{BN}(X%
\mathcal{)}$ stand for the maximal revenue using Bayesian Nash
implementation, i.e., mechanisms that satisfy IC-BN and IR-BN (in the
one-buyer case, i.e., when $n=1,$ these two concepts clearly coincide).

\bigskip

\noindent \textbf{Remarks.} \emph{(a) Bayesian Nash implementation for
independent buyers. }In the Bayesian Nash case, when the buyers' random
valuation vectors $X^{1},X^{2},...,X^{n}$ are \emph{independent}, only the
expectations of allocations and payments, conditional on each buyer's own
values, matter: replacing $q^{j}(x)$ with $\bar{q}^{j}(x^{j}):=\mathbb{E}%
\left[ q^{j}(X)|X^{j}=x^{j}\right] =\mathbb{E}\left[ q^{j}(x^{j},X^{-j})%
\right] $ and $s^{j}(x)$ with $\bar{s}^{j}(x^{j}):=\mathbb{E}\left[
s^{j}(X)|X^{j}=x^{j}\right] =\mathbb{E}\left[ s^{j}(x^{j},X^{-j})\right] $
throughout affects neither the IC-BN and IR-BN constraints nor the revenue.%
\footnote{%
However, the feasibility conditions $\sum_{i}q_{\ell }^{i}\leq 1$ on the
allocations \emph{cannot} be directly expressed in terms of the $\bar{q}^{i}$
(this is known as the \textquotedblleft implementability" condition; see
Border 1991 for a necessary and sufficient condition for implementability,
and Hart and Reny 2015b for a simple restatement and proof).} We will thus
assume without loss of generality that mechanisms in the Bayesian Nash case
are given in this \textquotedblleft reduced form" where $q^{j}$ and $s^{j}$
depend only on $x^{j}$ (rather than on the entire $x$), for all $j.$

\emph{(b) Non-positive transfer (NPT) and subdomain.} A mechanism $\mu
=(q^{j},s^{j})_{j=1,...,n}$ for $n\geq 1$ buyers and $k\geq 1$ goods
satisfies \textbf{NPT} if $s^{j}(x)\geq 0$ for every $j$ and every $x.$ The
results (i)--(iv) of Proposition \ref{p:sub-dom} in Section \ref{sus:model}
extend to multiple buyers, as follows.

In the dominant strategy case, we consider separately each $j=1,...,n$ and
each $x^{-j}\in \mathbb{R}_{+}^{k(n-1)},$ and obtain, in particular, that
NPT is equivalent to $s^{j}(0,x^{-j})=0$ for all $j$ and all $x^{-j},$ that
NPT can be assumed without loss of generality when maximizing revenue, and
that the subdomain property holds: $\mathbb{E}\left[ \sum_{j}s^{j}(X)\,%
\mathbf{1}_{X\in A}\right] \leq \,$\textsc{Rev}$^{DS}(X\,\mathbf{1}_{X\in
A})\leq \,$\textsc{Rev}$^{DS}(X)$ for every $A\subseteq \mathbb{R}_{+}^{kn}.$

In the Bayesian Nash case, when the buyers are independent (and each payment 
$s^{j}$ depends only on $x^{j}$; see Remark (a) above), we obtain in
particular that NPT is equivalent to $\bar{s}^{j}(0)=0$ for all $j,$ that
NPT can be assumed without loss of generality when maximizing revenue, and
that $\mathbb{E}\left[ \sum_{j}s^{j}(X)\,\mathbf{1}_{X\in A}\right] =\mathbb{%
E}\left[ \sum_{j}\bar{s}^{j}(X^{j})\,\mathbf{1}_{X\in A}\right] \leq \,$%
\textsc{Rev}$^{BN}(X\,\mathbf{1}_{X\in A})\leq \,$\textsc{Rev}$^{BN}(X)$ for
every\footnote{%
NPT need \emph{not} hold in the Bayesian Nash case when the buyers' random
valuations are \emph{not independent}. In this case, optimal mechanisms may
make use of negative payments $s^{j}(x)$ (i.e., positive transfers); cf. Cr%
\'{e}mer and McLean (1988). In addition, the requirement that $s^{j}$
depends only on $x^{j}$ is needed because the restriction to a set $A$ of
values of $X$ may introduce dependencies between the coordinates of $X\,%
\mathbf{1}_{X\in A},$ and then $\mathbb{E}\left[ \sum_{j}s^{j}(X)\,\mathbf{1}%
_{X\in A}\right] =\mathbb{E}\left[ \sum_{j}\bar{s}^{j}(X^{j})\,\mathbf{1}%
_{X\in A}\right] $ need not hold.} $A\subseteq \mathbb{R}_{+}^{kn}.$

\bigskip

We state the result separately for the two kinds of implementation, since in
the dominant strategy case the result is stronger: the requirement that the
buyers' random valuations are independent is not needed (i.e., while there
is independence between the two goods---$X_{1}^{j}$ and $X_{2}^{\ell }$ are
independent for any two buyers $j,\ell =1,...,n$---we allow, for each good $%
i=1,2,$ the buyers' values $X_{i}^{1},X_{i}^{2},...,X_{i}^{n}$ to be
arbitrarily correlated). The two theorems below give Theorem \ref{th:gfor-n}.

\begin{theorem}
\label{th:n-DS}In the case of $n$ buyers, two goods, and dominant strategy
implementation, if the random valuation vectors of the two goods $%
X_{1}=(X_{1}^{j})_{j=1,...,n}$ and $X_{2}=(X_{2}^{j})_{j=1,...,n}$ are
independent, then \textsc{GFOR}$($\textsc{separate}$)\geq 1/2,$ i.e.,%
\begin{equation*}
\text{\textsc{Rev}}^{DS}(X_{1})+\text{\textsc{Rev}}^{DS}(X_{2})\geq \frac{1}{%
2}\text{\textsc{Rev}}^{DS}(X_{1},X_{2}).
\end{equation*}
\end{theorem}

\begin{theorem}
\label{th:n-BN}In the case of $n$ independent buyers, two goods, and
Bayesian Nash implementation, if the random valuation vectors of the two
goods $X_{1}=(X_{1}^{j})_{j=1,...,n}$ and $X_{2}=(X_{2}^{j})_{j=1,...,n}$
are independent, then \textsc{GFOR}$($\textsc{separate}$)\geq 1/2,$ i.e.,%
\begin{equation*}
\text{\textsc{Rev}}^{BN}(X_{1})+\text{\textsc{Rev}}^{BN}(X_{2})\geq \frac{1}{%
2}\text{\textsc{Rev}}^{BN}(X_{1},X_{2}).
\end{equation*}
\end{theorem}

\begin{proof}[Proof of Theorems \protect\ref{th:n-DS} and \protect\ref%
{th:n-BN}]
Put $Y=X_{1}=(X_{1}^{j})_{j=1,...,n}$ and $Z=X_{2}=(X_{2}^{j})_{j=1,...,n}$
for the random valuation vectors of good $1$ and good $2$, respectively;
thus $Y$ and $Z$ are $\mathbb{R}_{+}^{n}$-valued random variables and $%
X=(Y,Z).$ Let $\mu =(q^{j},s^{j})_{j=1,...,n}$ be an NPT mechanism for the
two goods that satisfies either IC-DS and IR-DS, or IC-BN and IR-BN; in the
BN case, we assume in addition that it is in reduced form: $q^{j}$ and $%
s^{j} $ depend only on $x^{j}$ (cf. Remark (a) above).\footnote{%
Formally, put here $q^{i}(x):=\bar{q}^{i}(x^{i})$ and $s^{i}(x):=\bar{s}%
^{i}(x^{i});$ this allows the proof to apply mutatis mutandis to both
implementations, dominant strategy and Bayesian Nash.}

We split the total expected revenue from $\mu $ into two parts, according to
which one of\footnote{%
We write $a^{(1)}:=\max_{j=1,...,n}a^{j}$ for the maximal coordinate of a
vector $a=(a^{j})_{j=1,...,n}\in \mathbb{R}^{n}.$} $Y^{(1)}=\max_{j}Y^{j}$
and $Z^{(1)}=\max_{j}Z^{j}$ is higher, and show that%
\begin{eqnarray}
\mathbb{E}\left[ S(Y,Z)\mathbf{1}_{Y^{(1)}\geq Z^{(1)}}\right] &\leq &2\text{%
\textsc{Rev}}(Y)\text{\ \ and}  \label{eq:n-triangle} \\
\mathbb{E}\left[ S(Y,Z)\mathbf{1}_{Z^{(1)}\geq Y^{(1)}}\right] &\leq &2\text{%
\textsc{Rev}}(Z);  \label{eq:n-triangle2}
\end{eqnarray}%
adding the two inequalities yields our result (recall that $S\geq 0$ by NPT).

To prove (\ref{eq:n-triangle}) (from which (\ref{eq:n-triangle2}) follows by
interchanging $Y$ and $Z$), for every fixed vector of values $z\in \mathbb{R}%
_{+}^{n}$ of the $n$ buyers for the second good define a mechanism $(\hat{q},%
\hat{s})\equiv (\hat{q}^{z},\hat{s}^{z})$ for the first good by $\hat{q}%
^{j}(y):=q_{1}^{j}(y,z)$ and $\hat{s}^{j}(y):=s^{j}(y,z)-q_{2}^{j}(y,z)%
\,z^{j}$ for every $y\in \mathbb{R}_{+}^{n}$ and $j=1,...,n$, and put $\hat{S%
}(y):=\sum_{j}\hat{s}^{j}(y)$. The mechanism $(\hat{q},\hat{s})$ is IC and
IR for $y$, since $(q,s)$ was IC and IR for $(y,z)$ (for IC: only the
constraints $(\tilde{y}^{j},z^{j})$ vs. $(y^{j},z^{j})$ matter; for IR, $%
\hat{b}^{j}(y)=b^{j}(y,z)\geq 0$). Then $S(y,z)=\sum_{j}s^{j}(y,z)=\sum_{j}%
\hat{s}^{j}(y)+\sum_{j}z^{j}q_{2}^{j}(y,z)\leq \sum_{j}\hat{s}%
^{j}(y)+z^{(1)}=\hat{S}(y)+z^{(1)}$ (the inequality obtains because $0\leq
z^{j}\leq z^{(1)}$ and $\sum_{j}q_{2}^{j}\leq 1).$ Since $Y$ is independent
of $Z$ we get%
\begin{eqnarray*}
\mathbb{E}\left[ \,S(Y,Z)\mathbf{1}_{Y^{(1)}\geq Z^{(1)}}\,|\,Z=z\right] &=&%
\mathbb{E}\left[ \,S(Y,z)\mathbf{1}_{Y^{(1)}\geq z^{(1)}}\right] \\
&\leq &\mathbb{E}\left[ \hat{S}(Y)\mathbf{1}_{Y^{(1)}\geq z^{(1)}}\right] +%
\mathbb{E}\left[ z^{(1)}\mathbf{1}_{Y^{(1)}\geq z^{(1)}}\,\right] .
\end{eqnarray*}%
The first term is the revenue from a subdomain of values of $y,$ and so it
is at most the maximal revenue \textsc{Rev}$(Y)$ by Remark (b) above (in the
BN case, $\hat{s}^{j}$ depends only on $y^{j}$ since $s^{j}$ and $q^{j}$
depend only on $x^{j}$); as for the second term, 
\begin{equation}
\mathbb{E}\left[ z^{(1)}\mathbf{1}_{Y^{(1)}\geq z^{(1)}}\,\right] =z^{(1)}\,%
\mathbb{P}\left[ Y^{(1)}\geq z^{(1)}\right] \,\leq \text{\textsc{Rev}}(Y),
\label{eq:e(min)-le-rev(max)-n}
\end{equation}%
since posting a price of $z^{(1)}$ and giving the good $y$ to a buyer $j$
with $y^{j}\geq z^{(1)},$ if there is any, constitutes an IC and IR
mechanism for\footnote{%
As in the Proof of Theorem \ref{th:1/2} in Section \ref{s:independent}, we
are \emph{not} using the characterization of optimal mechanisms in the
one-good case (Myerson 1981), but only the fact that posting a price is IC
and IR.} $y$. Thus 
\begin{equation*}
\mathbb{E}\left[ \,S(Y,Z)\mathbf{1}_{Y^{(1)}\geq Z^{(1)}}\,|\,Z=z\right]
\leq 2\text{\textsc{Rev}}(Y)
\end{equation*}%
for every value $z$ of $Z$; taking expectation over $z$ yields (\ref%
{eq:n-triangle}).
\end{proof}

\subsection{Summary of Results\label{ap:summary}}

The two tables below summarize the results of this paper; $X$ stands for $%
(X_{1},X_{2},...,X_{k}),$ where $X_{1},X_{2},...,X_{k}$ are $k$ independent
goods (i.e., one-dimensional nonnegative random variables).\footnote{$%
\mathrm{o}(1)$ means \textquotedblleft converging to $0$ as $k\rightarrow
\infty $\textquotedblright ; see footnote \ref{ftn:O} for the $\mathrm{O,}$~$%
\Omega ,$ and $\Theta $ notations.} Table 1 provides the bounds on the
guaranteed fraction of optimal revenue for selling separately and for
selling as one bundle (with the four main results in bold). Table 2 provides
the comparisons between the separate and bundled revenues.\footnote{%
When comparing \textsc{SRev} and \textsc{BRev }we obtained tight bounds in
almost all cases (unlike the results for \textsc{GFOR}). This is in part due
to the fact that both \textsc{SRev} and \textsc{BRev} reduce to one-good
revenues, for which monotonicity holds (see Proposition \ref{p:1-monot} and
the extensive use of $\mathtt{ER}$ goods; cf. Remark (b) after Proposition %
\ref{p:ER}).}

{\renewcommand{\arraystretch}{1.8}\setlength{\tabcolsep}{7pt}\hspace{-0.75in}%
\begin{tabular}{||c||c|c|c|c||}
\hline\hline
{\setlength{\tabcolsep}{7pt}} & $k=2$ indep. & $k=2$ i.i.d. & $k\geq 2$
indep. & $k\geq 2$ i.i.d. \\ \hline\hline
$\forall X~~\dfrac{\text{\textsc{SRev}}(X)}{\text{\textsc{Rev}}(X)}\geq $ & $%
\mathbf{\dfrac{1}{2}}$ & $\mathbf{\dfrac{e}{e+1}\approx 0.73}$ & 
\multicolumn{2}{|c||}{$\mathbf{\Omega \left( \dfrac{1}{\mathrm{\mathbf{log}}%
^{2}k}\right) }$} \\ 
& \textbf{(Th.\ref{th:1/2})} & \textbf{(Th.\ref{th:73})} & 
\multicolumn{2}{|c||}{\textbf{(Th.\ref{th:srev-k})}} \\ \hline
$\exists X~~\dfrac{\text{\textsc{SRev}}(X)}{\text{\textsc{Rev}}(X)}\leq $ & 
\multicolumn{2}{c}{$\dfrac{1}{1+w}\approx 0.78$} & \multicolumn{2}{|c||}{$%
\mathrm{O}\left( \dfrac{1}{\log k}\right) $} \\ 
& \multicolumn{2}{c}{(Pr.\ref{p:gfor-le})} & \multicolumn{2}{|c||}{(Co.\ref%
{c:sep-k-er})} \\ \hline\hline
$\forall X~~\dfrac{\text{\textsc{BRev}}(X)}{\text{\textsc{Rev}}(X)}\geq $ & $%
\dfrac{1}{2}\cdot \dfrac{1}{2}=\dfrac{1}{4}$ & $\dfrac{e}{e+1}\cdot \dfrac{2%
}{3}$ & $\Omega \left( \dfrac{1}{k}\right) $ & $\mathbf{\Omega \left( \dfrac{%
1}{\mathrm{\mathbf{log}}\;k}\right) }$ \\ 
& (Th.\ref{th:1/2}+Pr.\ref{p:bun>sep}(i)) & (Th.\ref{th:73}+Le.\ref{l:2/3})
& (Le.\ref{B-1/k}) & \textbf{(Th.\ref{th:brev-k})} \\ \hline
$\exists X~~\dfrac{\text{\textsc{BRev}}(X)}{\text{\textsc{Rev}}(X)}\leq $ & $%
\dfrac{1}{2}+\varepsilon $ & $\dfrac{2}{3}$ & $\dfrac{1}{k}+\varepsilon $ & $%
\approx 0.57+\mathrm{o}(1)$ \\ 
& (Ex.\ref{ex:1/k}) & (Ex.\ref{ex:2/3}) & (Ex.\ref{ex:1/k}) & (Ex.\ref{ex:57}%
) \\ \hline\hline
\end{tabular}%
}

\begin{center}
\textbf{Table 1}. Summary of results for \textsc{GFOR}
\end{center}

{\renewcommand{\arraystretch}{1.8}\setlength{\tabcolsep}{7pt}\hspace{-1in} 
\begin{tabular}{||c||c|c|c|c||}
\hline\hline
{\setlength{\tabcolsep}{7pt}} & $k=2$ indep. & $k=2$ i.i.d. & $k\geq 2$
indep. & $k\geq 2$ i.i.d. \\ \hline\hline
$\inf_{X}\dfrac{\text{\textsc{SRev}}(X)}{\text{\textsc{BRev}}(X)}=$ & 
\multicolumn{2}{c}{$\dfrac{1}{1+w}\approx 0.78$} & \multicolumn{2}{|c||}{$%
\mathrm{\Theta }\left( \dfrac{1}{\log k}\right) $} \\ 
& \multicolumn{2}{c}{(Pr.\ref{p:sep>bun}(i)+Pr.\ref{p:ER}(iii))} & 
\multicolumn{2}{|c||}{(Pr.\ref{p:sep>bun}(ii)+Pr.\ref{p:ER}(iv))} \\ \hline
$\inf_{X}\dfrac{\text{\textsc{BRev}}(X)}{\text{\textsc{SRev}}(X)}=$ & $%
\dfrac{1}{2}$ & $\dfrac{2}{3}$ & $\dfrac{1}{k}$ & $\in \left[ \dfrac{1}{4}%
,0.57+\mathrm{o}(1)\right] $ \\ 
& (Pr.\ref{p:bun>sep}(i)+Ex.\ref{ex:1/k}) & (Le.\ref{l:2/3}+Ex.\ref{ex:2/3})
& (Pr.\ref{p:bun>sep}(i)+Ex.\ref{ex:1/k}) & (Pr.\ref{p:bun>sep}(ii)+Ex.\ref%
{ex:57}) \\ \hline\hline
\end{tabular}%
}

\begin{center}
\textbf{Table 2.} Summary of results for \textsc{SRev} vs. \textsc{BRev}
\end{center}

\end{document}